\documentclass[lettersize,journal]{IEEEtran}
\usepackage{amsmath,amsfonts}
\usepackage{algorithm}
\usepackage{array}
\usepackage[caption=false,font=normalsize,labelfont=sf,textfont=sf]{subfig}
\usepackage{textcomp}
\usepackage{stfloats}
\usepackage{url}
\usepackage{verbatim}
\usepackage{graphicx}
\usepackage{cite}

\usepackage{amssymb}
\usepackage{verbatim}
\usepackage{graphicx}
\usepackage{algorithmicx}
\usepackage{algpseudocode}
\usepackage{multirow}
\usepackage{makecell}
\usepackage{enumitem}
\usepackage{xcolor}
\usepackage{multirow}
\usepackage{multicol}
\usepackage{booktabs}
\usepackage{adjustbox}
\usepackage{hyperref}
\usepackage{amsthm}
\usepackage{bbm}
\newtheorem{theorem}{Theorem}
\newtheorem{proposition}{Proposition}
\newtheorem{lemma}{Lemma}

\newtheorem{definition}{Definition}

\begin{document}

\title{Learning-based Privacy-Preserving Graph Publishing Against Sensitive Link Inference Attacks}


\author{Yucheng Wu$^\dag$, Yuncong Yang$^\dag$, Xiao Han, Leye Wang, Junjie Wu
\thanks{\textit{$^\dag$Equal contribution.}}
\thanks{Manuscript received 2 December 2024; revised 6 June 2025 and 7 July 2025; accepted 22 July 2025. Date of publication XX August 2025; date of current version XX August 2025. 
The authors thank the senior editor, associate editor, and anonymous reviewers for their guidance and constructive comments that have tremendously improved the paper.  X. Han and J. Wu are supported in part by the National Key R\&D Program of China (2023YFC3304700). X. Han is supported in part by the National Natural Science Foundation of China under grant No. 72071125 and 72031001. J. Wu is supported in part by the National Natural Science Foundation of China under grant No. 72031001, 72242101, and the Outstanding Young Scientist Program of Beijing Universities (JWZQ20240201002). \textit{(Corresponding author: Xiao Han.)}}
\thanks{Yucheng Wu and Leye Wang are with the Key Lab of High Confidence Software Technologies, Peking University, Ministry of Education, Beijing 100871, China, and the School of Computer Science, Peking University, Beijing 100871, China (e-mail: wuyucheng@stu.pku.edu.cn; leyewang@pku.edu.cn).}
\thanks{Yuncong Yang is with Key Laboratory of Interdisciplinary Research of Computation and Economics (Shanghai University of Finance and Economics), Ministry of Education, Shanghai 200433, China, and the School of Information Management and Engineering, Shanghai University of Finance and Economics, Shanghai 200433, China (e-mail: yycphd@163.sufe.edu.cn).}
\thanks{Xiao Han and Junjie Wu are with the Key Laboratory of Data Intelligence and Management, Beihang University, Ministry of Industry and Information Technology, Beijing 100191, China, and the School of Economics and Management, Beihang University, Beijing 100191, China (e-mail: xh\_bh@buaa.edu.cn; wujj@buaa.edu.cn)}
\thanks{This paper has supplementary downloadable material available at \href{http://ieeexplore.ieee.org.}{http://ieeexplore.ieee.org.}, provided by the author. The material includes proof of theoretical analysis and additional experimental results. Contact wuyucheng@stu.pku.edu.cn for further questions about this work.}
}

\markboth{IEEE Transactions on Information Forensics and Security,~Vol.~xx, No.~x, August~2025}%
{Wu \MakeLowercase{\textit{et al.}}: Learning-based Privacy-Preserving Graph Publishing Against Sensitive Link Inference Attacks}

\IEEEpubid{0000--0000/00\$00.00~\copyright~2021 IEEE}

\maketitle

\begin{abstract}
Publishing graph data is widely desired to enable a variety of structural analyses and downstream tasks. However, it also potentially poses severe privacy leakage, as attackers may leverage the released graph data to launch attacks and precisely infer private information such as the existence of hidden sensitive links in the graph. Prior studies on privacy-preserving graph data publishing relied on heuristic graph modification strategies and it is difficult to determine the graph with the optimal privacy--utility trade-off for publishing. In contrast, we propose the first \textit{privacy-preserving graph structure learning framework against sensitive link inference attacks}, named PPGSL, which can automatically learn a graph with the optimal privacy--utility trade-off. The PPGSL operates by first simulating a powerful surrogate attacker conducting sensitive link attacks on a given graph. It then trains a parameterized graph to defend against the simulated adversarial attacks while maintaining the favorable utility of the original graph. To learn the parameters of both parts of the PPGSL, we introduce a secure iterative training protocol. It can enhance privacy preservation and ensure stable convergence during the training process, as supported by the theoretical proof. Additionally, we incorporate multiple acceleration techniques to improve the efficiency of the PPGSL in handling large-scale graphs. The experimental results confirm that the PPGSL achieves state-of-the-art privacy--utility trade-off performance and effectively thwarts various sensitive link inference attacks.
\end{abstract}


\section{Introduction}
\label{sec:intro}
\IEEEPARstart{G}{raph} data are ubiquitous in our daily lives, spanning the realms of social relationships~\cite{han2015alike,fan2019graph}, communication networks~\cite{suarez2022graph}, and traffic networks~\cite{wang2021exploring}, \textit{etc}. Owing to the abundant information in graph data, it is common practice for data holders to publish them for academic and economic benefits. For example, universities and research institutes such as SNAP~\cite{snapnets} and AMiner~\cite{aminer} collect and release substantial volumes of graph data, significantly fostering the development of graph data mining; social media share their data via open APIs (\textit{e.g.}, Facebook~\cite{facebook}) for business. 
However, sharing graph data without adequate protection may result in severe privacy leakage problems, especially when encountering various inference attacks~\cite{jia2017attriinfer,zhang2020towards,duddu2020quantifying}. 
According to privacy laws, including GDPR~\cite{gdpr}, it is imperative to protect private information so that users are unwilling to be exposed to data publications~\cite{yao2019sensitive}. This leads to an urgent need for privacy protection measures when sharing graph data.

This work focuses on one of the most common inference attacks on graphs: \textit{sensitive link inference attacks}~\cite{han2023privacy,he2021stealing,yu2019target}. These attacks use various techniques on released graph data to accurately deduce hidden sensitive links between users, posing significant risks to their private information. Hidden sensitive links, such as private friendships and confidential transactions, are the links that users intentionally conceal to protect their privacy and remain invisible to the public. For example, the Facebook platform allows its users to make their partial friendships private and invisible from others, and these hidden friendships are regarded as sensitive links. Despite the invisibility of sensitive links in the published graph, they can still be inferred because of the pronounced similarity between their connected node pairs (\textit{i.e.}, sensitive node pairs).
Preliminary empirical analyses show that node pairs with hidden links exhibit significantly greater similarity across various metrics than do unlinked node pairs in Table~\ref{tbl:structure_proximity}. This suggests that attackers can easily infer the presence of hidden links by evaluating the similarity between unconnected node pairs.

\begin{table}[t]
  \caption{Structural proximity (including average shortest path length and disconnection ratio), attribute cosine similarity and embedding cosine similarity (where the embedding is produced by GraphSAGE~\cite{hamilton2017inductive}) of different types of node pairs}
  \label{tbl:structure_proximity}
  \begin{adjustbox}{max width=1\linewidth}
  \centering
  \begin{tabular}{ccccccc}
    \toprule
    Dataset &\makecell{Node pair\\ type} &\makecell{Avg. shortest \\path length} &\makecell{Discon. \\ratio} &\makecell{Attribute \\similarity} &\makecell{Embedding\\ similarity} \\
    \midrule
    \multirow{3}{*}{Cora}
    &w/ visible link  &1.000  &0.000  &0.011  &0.900 \\
    &w/ hidden link &3.206  &0.164  &0.011  &0.702 \\
    &w/o link &7.007  &0.254  &0.003  &0.499 \\
    \midrule
    \multirow{3}{*}{LastFMAsia}
    &w/ visible link  &1.000  &0.000  &0.010  &0.978 \\
    &w/ hidden link &2.480  &0.089  &0.010  &0.947 \\
    &w/o link &5.560  &0.124  &0.008  &0.820 \\
    \bottomrule
    \end{tabular}
\end{adjustbox}
\vspace{-1em}
\end{table}

\IEEEpubidadjcol

To defend against sensitive link inference attacks, a straightforward strategy is to reduce the structural proximity
of sensitive node pairs, which applies to both attributed and unattributed graphs. However, this strategy unavoidably introduces disturbance to the original graph, thereby leading to a degradation of graph data utility. Therefore, researchers strive to improve the protection of sensitive link privacy while minimizing the loss of graph utility, \textit{i.e.}, seeking a graph that strikes an optimal privacy--utility trade-off for publishing. 
Concerning the computational infeasibility of enumerating all potential graph structures, they resort to heuristic graph modification solutions. More specifically, they often randomly select a small set of node pairs from the original graph and greedily identify the best link modifications (\emph{e.g.}, adding or deleting links) among these candidate node pairs for a favorable privacy--utility trade-off~\cite{han2023privacy,yu2019target}.
Nevertheless, the main limitation of these heuristic solutions is the lack of guarantees for reaching the optimal privacy--utility trade-off, which makes them prone to becoming stuck in local optima. To overcome the limitations of existing methods, we aim to develop a solution that can derive an optimal graph structure for release within a reasonable computational time.

Recent advances in graph structure learning (GSL) have achieved remarkable success in solving graph-related tasks~\cite{zhu2021survey}. 
GSL treats the adjacency matrix of a graph as a set of continuous parameters, enabling the application of optimization techniques to refine the parameterized matrix and determine the optimal graph structure for specific objectives.
However, most existing GSL methods focus on designing objective functions that enhance robustness, smoothness, and task performance~\cite{jin2020graph,liu2022compact}, with limited attention given to privacy protection.
\textbf{This work proposes and investigates the problem of privacy-preserving graph structure learning, aiming to identify the optimal graph structure that achieves the best privacy--utility trade-off for graph publishing}. To address this problem, we face the following challenges:

\textit{Challenge 1: Differentiable and universal privacy protection objective.} It is crucial to design an appropriate objective to navigate the graph learning process toward protecting sensitive link privacy.
While existing privacy-related studies have proposed some privacy protection objectives, they are typically nondifferentiable and thus unsuitable for graph structure learning~\cite{yu2019target}.
How can we create a differentiable privacy objective that provides a supervisory signal throughout the learning process? In addition, how can we guarantee that the learned graph, guided by this objective, is robust enough to withstand a variety of sensitive link inference attacks?

\textit{Challenge 2: Optimal privacy--utility trade-off.} Solely focusing on privacy protection can severely degrade the utility of the learned graph, thereby rendering graph-based tasks ineffective. How can we design a utility objective function that ensures that the learned graph retains as much utility as possible from the original graph? Furthermore, how
can we balance privacy and utility to jointly learn a graph that maximally preserves privacy while minimizing utility loss?

\textit{Challenge 3: Effective and efficient training protocol.}
The graph learning process is complex and involves multiple objectives and many trainable parameters (especially in large graphs).
While most GSL methods adopt an alternating or end-to-end training protocol for the components corresponding to different objectives~\cite{zhu2021survey}, these protocols may expose privacy during training and lack guarantees of stable convergence.
How can we design an effective and efficient training protocol to avoid suboptimal solutions and unstable convergence while ensuring high efficiency and scalability?

By jointly considering the above challenges, this work makes the following contributions:
\begin{itemize}[leftmargin=5.5mm]
    \item To the best of our knowledge, we are the first to formalize the problem of learning a privacy-preserving graph against sensitive link inference attacks, which aims to learn graphs for publishing with the optimal privacy--utility trade-off. Our research broadens both the realms of privacy-preserving and graph structure learning studies.
    \item We propose a generic \textbf{P}rivacy-\textbf{P}reserving \textbf{G}raph \textbf{S}tructure \textbf{L}earning (\emph{i.e.}, \textbf{PPGSL}) framework, which is designed to automatically learn an optimal graph structure that protects sensitive link privacy.
    Within the PPGSL, we develop an effective and differentiable privacy protection objective derived by establishing a surrogate attack model and conducting inference on sensitive links (tackling Challenge 1).
    We also devise a utility objective that minimizes the distortion between the learned and original graphs, guiding the learning process toward an optimal privacy--utility balance alongside the privacy protection objective (tackling Challenge 2).
    To ensure stable convergence, we design a secure iterative training protocol (\emph{i.e.}, SITP) and introduce speed-up strategies to further increase scalability and efficiency (tackling Challenge 3).
    \item Rigorous theoretical analyses confirm the key benefits of PPGSL with SITP, including stable convergence, optimal privacy--utility trade-off, and universal defense capability.
\end{itemize}

We conduct extensive empirical evaluations on six real-world graph datasets across two common utility tasks (node classification and link prediction). Our results show that the PPGSL achieves a superior privacy--utility trade-off compared with baseline methods and effectively defends against various sensitive link inference attacks.\footnote{The code is available at \href{https://github.com/wuyucheng2002/PPGSL}{https://github.com/wuyucheng2002/PPGSL}, along with detailed parameter settings and experimental results of the PPGSL.}

\section{Problem Formulation}
\subsection{Preliminaries}
Let $G=(V,E)=(A,X)$ be a graph, where $V$ is the set of nodes and $E \subseteq V \times V$ represents the set of publicly observable links on the graph. $N=|V|$ is the total number of nodes. For the unweighted graph, we use $A=\{w_{ij}\} \in \{0,1\}^{N \times N}$ to denote the adjacency matrix of $G$, where the entry $w_{ij}=1$ if the link $\left \langle v_i,v_j \right \rangle \in E$; otherwise, $w_{ij}=0$. For the weighted graph, $A=\{w_{ij}\} \in \mathbb R^{N \times N}$, where $w_{ij}$ is the edge weight of $\left \langle v_i,v_j \right \rangle$. $X \in \mathbb{R}^{N \times D}$ denotes the node feature matrix, where $D$ is the dimension of the node features.
Let $G'=(V,E')=(A',X')$ denote the released graph data\footnote{Note that we do not add or remove nodes on the graph $G$ (\textit{i.e.}, $V'=V$), so we omit $V'$ here.}.

\begin{definition}[Sensitive links]
The set of sensitive links, denoted as $E_s \subseteq V \times V$, exists yet remains unobservable on graph $G$ due to privacy concerns, signifying that $E_s \cap E = \varnothing$.
\end{definition}

\begin{definition}[Sensitive node pair]
A pair of nodes $v_i$ and $v_j$ is called a sensitive node pair if $\left \langle v_i,v_j \right \rangle\in E_s$.
\end{definition}

Furthermore, we define \textit{nonexistent links} as pairs of nodes that are not connected by any link (neither a publicly observable nor sensitive link) in the original graph, denoted by $E_n$. Thus, we have $E_n \cap E_s = \varnothing$ and $E_n \cap E = \varnothing$.

\subsection{Sensitive Link Inference Attacks}
\label{sec:threat_model}
\textbf{Attackers' goal.} Attackers aim to accurately infer the existence of sensitive links $E_s$ within the graph $G$.
Taking Facebook as an example, it shares certain user attributes (\textit{e.g.}, nickname and age) and public friendship data via an open API, typically for commercial purposes. However, this data sharing also poses a threat to the privacy of users who choose to hide their private friendships from a public view. Specifically, malicious data users (\textit{i.e.}, attackers) can retrieve the public social network data via the API and leverage it to conduct inference attacks on the hidden private friendships between the privacy-sensitive users.
Leaks of private friendships not only violate user privacy but can also have more severe consequences. For example, attackers may exploit hidden links to defraud or blackmail users or execute advanced network attacks such as Phishing attacks~\cite{alkhalil2021phishing} and Sybil attacks~\cite{fong2011preventing}.

\textbf{Attackers' knowledge.} Attackers can access the released graph \( G' \) on public platforms but do not have any information on the sensitive links \( E_s \).

\textbf{Attack model.} Attackers can employ various algorithms to establish an attack model (denoted as $\mathcal{M}$), such as employing network embedding similarity to infer link existence~\cite{grover2016node2vec}. Note that the true sensitive links are unavailable to attackers, and the only available information is $G'$. Therefore, they may sample a set of existing/nonexistent links from $G'$ as positive/negative samples to construct a dataset used to train the attack model $\mathcal M$ and then conduct inference on $E_s$ via the well-trained model.
This is reasonable because sensitive pairs are expected to exhibit high similarity, such as connected node pairs (as demonstrated in Sec.~\ref{sec:intro}).
Formally, given $G'$ and the chosen attack model $\mathcal{M}$, the training process aims to minimize the expected inference error on $E_{p}'$ and $E_{n}'$, where $E_{p}' \subset E'$ denotes the existing links sampled from $G'$ and $E_{n}'$ represents the nonexistent links sampled from $G'$:
\begin{equation}
    \mathcal{M}^{*} = \arg\min_{\mathcal{M}}(error (\mathcal{M}(G'), E_{p}'\cup E_{n}'))
\end{equation}
After the training process of the attack model, attackers utilize the well-trained model $\mathcal M^{*}$ to infer the existence of $E_s$.

\subsection{Problem Definition}

Now, we define our \textbf{privacy-preserving graph structure learning problem} as follows: Given an original graph $G$ and a set of sensitive links $E_s$, our objective is to learn a graph $G'$ via a graph structure learning function $\mathcal H_\theta (\cdot)$ with a set of parameters $\theta$ such that $G'=\mathcal H_\theta (G)$. The learned graph~$G'$ should maximize the error of the well-trained attack model $\mathcal M^*$ in inferring the sensitive links $E_s$ to protect privacy; $G'$ should also minimize the data distortion for utility preservation. Mathematically, our problem can be formulated as follows:
\begin{align}
    \textbf{Privacy goal:} \ &\max_{\theta}(error(\mathcal{M}^{*}(\mathcal H_\theta (G)), E_s)),
    \label{eq:privacy_goal}\\
    \mathcal{M}^{*} = &\arg\min_{\mathcal{M}}(error (\mathcal{M}(\mathcal H_\theta (G)), E_{p}'\cup E_{n}'))\\
    \textbf{Utility goal:} \ &\min_\theta (dist(\mathcal H_\theta (G), G)) \label{eq:utility_goal}
\end{align}
where $dist(\cdot,\cdot)$ is a general distance function used to quantify the difference between the topologies of two graphs. 

\section{Proposed Framework: PPGSL}
In this section, we elucidate our proposed privacy-preserving graph structure learning (\emph{i.e.}, PPGSL) framework, which learns a graph with the optimal privacy--utility trade-off against sensitive link inference attacks. In particular, the PPGSL consists of two iterative training modules (\textit{i.e.}, \textit{surrogate attack module} and \textit{privacy-preserving graph learner module}), and its overview is depicted in Fig.~\ref{fig:overview}. The \textit{surrogate attack module} imitates an attacker conducting inference attacks in two phases. First, it constructs a surrogate attack model that is based on a given learned graph; second, it uses the attack model to infer the existence of sensitive links and induces a privacy leakage risk due to the current graph. The privacy leakage information is subsequently delivered to the other graph learner module as a supervisory signal of the learning process.
In addition, the \textit{privacy-preserving graph learner module} parameterizes the graph structure to be trainable and manages the learning process. It updates the graph structure to reduce privacy leakage while controlling the distortion to maintain data utility.
We also propose a \textit{secure iterative training protocol} to iteratively train the two modules, enhancing both the privacy level of the learned graph and the stability of the optimization process.
In addition, several \textit{speed-up strategies} are put forth to improve the scalability and efficiency of the PPGSL.

\begin{figure}[t]
  \centering
  \includegraphics[width=\linewidth]{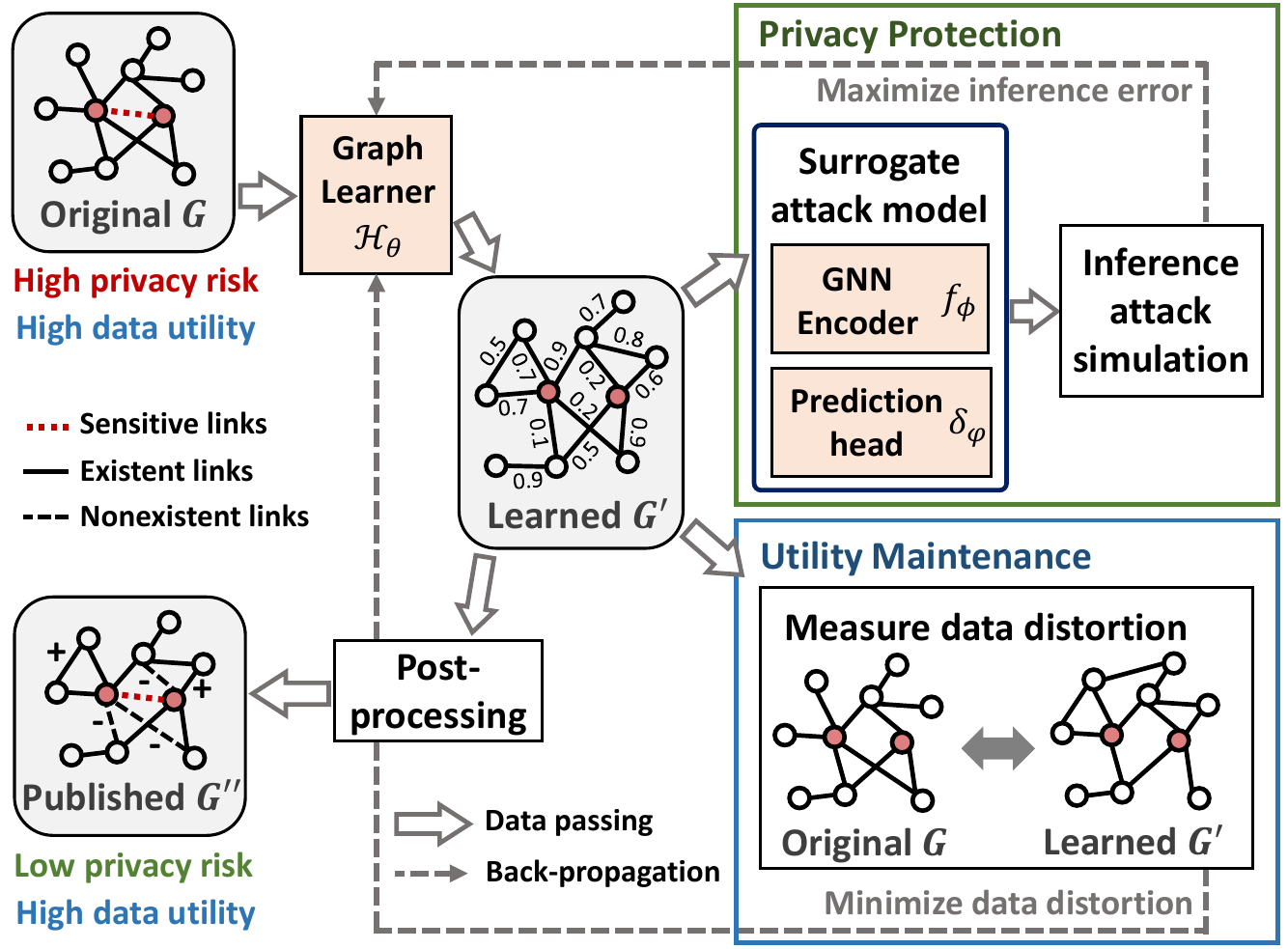}
  \vspace{-1em}
  \caption{Overview of the PPGSL. 
  }
  \vspace{-1em}
  \label{fig:overview}
\end{figure}

\subsection{Surrogate Attack Module}
\label{sec:privacy}
To guide our graph learner in producing a privacy-preserving graph, we need to provide supervisory signals that direct the learned graph structure toward minimizing privacy leakage. To achieve this goal, we first establish a surrogate attack model to simulate inference attacks on sensitive links. The privacy leakage risk induced by this surrogate attack model then serves as a supervisory signal, instructing the graph learner to produce a graph that resists inference attacks.

Previous studies have proposed the use of heuristic metric-based surrogate attack models, such as the resource allocation index~\cite{yu2019target}, for similar purposes. However, these approaches face two significant limitations. First, simple heuristics capture limited information, providing only local structure proximity details, which may result in weak surrogate attack models and imprecise supervisory signals. Second, existing heuristic metric-based surrogate attack models are nondifferentiable, making them unsuitable for providing signals that support gradient back-propagation in our learning-based graph structure optimization process.

To address these issues, we propose a machine learning (ML)-based surrogate attack model. As described in Sec.~\ref{sec:intro}, the relatively high structural proximity between sensitive node pairs significantly contributes to the ease with which attackers infer sensitive links. By leveraging ML techniques, we develop a surrogate attack model that effectively exploits structural proximity information within the graph. Furthermore, this ML-based surrogate attack model ensures that both the training and inference processes are fully differentiable.

\subsubsection{Surrogate Attack Model Architecture}
Our surrogate attack model comprises a graph neural network (GNN) encoder and a prediction head. 
Its forward propagation 
is differentiable, allowing us to utilize its information to direct the optimization of graph structure learning through gradient back-propagation. Moreover, 
GNN encoders are widely used by attackers and present effective attack performance in previous works~\cite{he2021stealing}, making our surrogate attack model highly representative.

In particular, we employ a GNN encoder (denoted as $f_\phi$) to generate informative node embeddings, as GNNs are highly expressive for representing graph data~\cite{xu2018how}. As analyzed in Sec.~\ref{sec:intro}, node pairs with sensitive links exhibit increased structural proximity, a vulnerability often exploited by attackers. The GNN effectively identifies this relationship---nodes with close structural distance will produce embeddings with high similarity as the GNN iteratively aggregates neighborhood information. Generally, any GNN model, such as GCN~\cite{kipf2016semi} or GraphSAGE~\cite{hamilton2017inductive}, can serve as the encoder.

Additionally, the prediction head (denoted as $\delta_\varphi$) determines the existence of a link on the basis of the learned node embeddings of its two ends. It uses a similarity metric such as cosine similarity or inner product~\cite{he2021stealing} or employs multilayer perceptrons (MLPs) that concatenate the embeddings of the target link's two ends as input vectors and output the similarity score of the link. The prediction head further includes a sigmoid function $\sigma (x)=1/(1+e^{-x})\in (0,1)$ to convert the similarity score into the predicted probability of a link's existence.
Specifically, a probability near 0 suggests dissimilar node embeddings and the absence of the target link, whereas a probability close to 1 indicates similar node embeddings and the presence of the target link.
The GNN encoder $f_\phi$ and prediction head $\delta_\varphi$ together form an effective and representative attack method, which substitutes $\mathcal M$ in Eq.~\ref{eq:privacy_goal}.

\subsubsection{Surrogate Attack Model Training Objective}
Essentially, our surrogate attack model is trained to infer the existence of links on the basis of node embedding similarity in the learned graph. Given a learned graph $G'=\mathcal H_\theta(G)$, we first construct the training dataset for the surrogate attack model by sampling an edge set $E_{p}'$ from $G'$ as the positive samples and a set of nonexistent edges $E_{n}'$ as the negative samples. Then, we train the surrogate attack model by making the predicted link existence probability close to the edge weight for positive samples and close to 0 for negative samples. The objective is as follows:
\begin{equation}
\label{eq:gnn_encoder}
\begin{aligned}
    \min_{\phi,\, \varphi} \mathcal{L}_{attack} \, =\, & \mathbb E_{\langle v_i,v_j\rangle \in E_{p}'}\, \textit{CE} \,(\delta_\varphi(z'_i,z'_j),w'_{ij}) \\
    &+ \mathbb E_{\langle v_i,v_j\rangle \in E_{n}'}\, \textit{CE} \,(\delta_\varphi(z'_i,z'_j),0)
\end{aligned}
\end{equation}
where $\textit{CE}\,(\cdot, \cdot)$ is the cross-entropy function. We use $Z'=f_{\phi}(G')$ to denote the learned node representations on graph $G'$, $z'_i$ denotes the representation of node $v_i$, and $w'_{ij}$ is the weight of edge $\langle v_i,v_j\rangle$. Intuitively, using Eq.~\ref{eq:gnn_encoder}, connected nodes are forced to have similar embeddings (the degree of similarity relates to the 
edge weight), whereas unconnected weights are repelled to exhibit orthometric embeddings.

\subsubsection{Inference Attack Simulation}
Once the training of the surrogate attack model is complete, it can be used to infer the existence of sensitive links. By inputting each sensitive link into the model, we obtain a predicted probability indicating its existence. Essentially, the simulation results provide an estimate of the privacy leakage of sensitive links. To precisely estimate the risk, we ensure that the surrogate attack model is adequately trained for convergence.

\subsection{Privacy-Preserving Graph Learner Module}
The privacy-preserving graph learner is a nonlinear learning model that transforms an original graph into a new privacy-preserving graph. It comprises a graph model architecture and learning objectives that address both privacy protection and utility maintenance goals. We will subsequently elaborate on these components.

\subsubsection{Graph Model Architecture}
Like various graph model architectures in recent studies for other graph learning purposes, such as graph denoising~\cite{fatemi2021slaps,liu2022towards}, we apply a straightforward full graph parameterization (FGP) approach to construct the graph model architecture for our privacy-preserving aim. FGP treats each entry of the adjacency matrix of a graph as an independent parameter and allows the learning of any adjacency matrix. We use $\mathcal H$ with parameters $\theta$ to denote the graph model architecture using FGP, \textit{i.e.}, $G'=\mathcal H_\theta(G)$.
Since $\theta$ may encompass nonsymmetrical and negative values, which are impermissible for an adjacency matrix, we refine $\theta$ to derive the learned adjacency matrix with symmetrization and truncation techniques. In particular, symmetrization is executed as $\theta^{s} = ( \theta + \theta^{\top})/2 $. Additionally, we confine the values of $\theta^{s}_{ij}$ between 0 and 1, truncating those falling below 0 or exceeding 1, $\theta^{st}_{ij}=\min\{\max\{\theta^s_{ij},0\},1\}\in [0,1]$.
In summary, $\mathcal H_\theta (G)$ can be determined as follows:
\begin{equation}
    \mathcal H_\theta (G) = \min\left\{\max\left\{\frac{\theta + \theta^{\top}}{2},0\right\},1\right\}
\end{equation}
As a result, the adjacency matrix $A'=\mathcal H_\theta (G)\in [0,1]^{N\times N}$ is continuous,\footnote{Since we only learn the graph structure and do not generate new node attributes, \textit{i.e.}, $G'=(A',X)$, we also use the expression $A'=\mathcal H_\theta(G)$ for brevity.} and we can interpret its term $w'_{ij}$ as the weight of edge $\langle v_i,v_j\rangle$ throughout the training phase. Notably, while we use FGP here for its simplicity and flexibility, the proposed PPGSL framework can accommodate alternative graph model architectures in place of FGP.

\subsubsection{Privacy Protection Objective}
\label{sec:privacy_protection_objective}
As described in Sec.~\ref{sec:privacy}, given a graph $G'$ produced by graph learner, the surrogate attack module trains an attack model and simulates inference attacks on sensitive links. The results of the attack serve as an estimate of the privacy leakage of sensitive links in the produced graph. To guide the graph learner toward reducing privacy leakage, we use the inference results of the surrogate attack model on the current learned graph as a supervisory signal. Specifically, we aim to update the graph structure so that the surrogate attack model is more likely to misclassify sensitive links as nonexistent, thereby achieving the privacy protection goal. This objective is formulated as:
\begin{equation}
\label{eq:embedding_dis}
\begin{aligned}
    \min_{\theta} \mathcal L_{priv} = \mathbb E_{\left \langle v_i,v_j \right \rangle \in E_s} \,\textit{CE} \,(\delta_\varphi(z'_i,z'_j), \, 0)
\end{aligned}
\end{equation}

\subsubsection{Utility Maintenance Objective}
It is also important to maintain the utility level of graph data for downstream applications. As delineated in Eq.~\ref{eq:utility_goal}, we measure the loss of data utility by the distortion between the original graph and the learned graph.
The original graph data are considered to possess the highest level of data utility, and any data distortion will lead to a deterioration in utility, as noted in prior studies~\cite{yang2016privcheck,jia2018attriguard,hsieh2021netfense}. In essence, more pronounced data distortion indicates greater loss in data utility.
Specifically, the data distortion can be ascertained via distance metrics, \textit{e.g.}, the Frobenius norm.
We aspire to guide the learned graph structure toward minimal distortion; hence, the utility maintenance objective function can be instantiated as follows:
\begin{equation}
\label{eq:utility_F}
    \min_\theta \mathcal L_{util} = {\Vert A - A' \Vert}^{2}_{F}
\end{equation}
Optimizing Eq.~\ref{eq:utility_F} guarantees that the learned adjacency matrix $A'$ aligns closely with the original adjacency matrix $A$, thereby preserving the utility of the learned graph data. Since this objective provides a universal gauge and performs well across various downstream tasks, it is particularly suitable for graph publishing scenarios, where downstream tasks are unknown.

\subsubsection{Overall Objective}
\label{sec:overall_objective}
To derive the overall objective for optimizing the graph learner, the privacy protection and utility maintenance objectives can be combined as follows:
\begin{equation}
\label{eq:overall_loss}
    \min_\theta \mathcal L_{learner} = \mathcal L_{priv} + \alpha \mathcal L_{util}
\end{equation}
where $\alpha\in [0, +\infty)$ functions as a hyperparameter to control the trade-off between the privacy protection effect and the data utility level. A smaller $\alpha$ indicates intensified privacy protection, albeit at the expense of reduced data utility.

\subsubsection{Postprocessing} If the original graph $G$ is unweighted, we need to recover the learned adjacency matrix $A'=\mathcal H_\theta(G)$ with continuous values into an unweighted adjacency matrix $\Tilde{A}$ with \textit{discrete} values for publishing. To this end, we employ a postprocessing method using an independent Bernoulli distribution $\mathcal{T}$ with the edge weight $A'_{ij}$ as the probability parameter to discretize the learned adjacency matrix, \textit{i.e.}, $A''_{ij}=\mathcal{T}(A'_{ij})$. Thus, the final published adjacency matrix $A''$ possesses entries of either 0 or 1.

\subsection{Secure Iterative Training Protocol}
\label{sec:training_pro}
Analytically, it is challenging to simultaneously optimize the two parameter groups in our framework's surrogate attack model and privacy-preserving graph learner. To address this, we introduce an innovative secure iterative training protocol (denoted as SITP) that iteratively trains both models. Unlike typical adversarial regularization approaches (\textit{e.g.}, AdvReg~\cite{nasr2018machine}), which alternate training between the two models until they converge simultaneously, SITP focuses on training the privacy-preserving graph learner by using a well-trained surrogate attack model at each step of the gradient descent process to update the graph learner's parameters $\theta$.
In other words, SITP iteratively retrains a surrogate attack model and performs a gradient descent update for graph structure learning until the privacy-preserving graph learner converges.
Formally, considering that the two sets of variables $\phi$ (of $f_\phi$) and $\theta$ (of $\mathcal H_\theta$) in the PPGSL\footnote{$\varphi$ (of $\delta_\varphi$) can be seen as part of $\phi$ (of $f_\phi$), as $\varphi$ and $\phi$ are updated jointly.}, SITP performs the following optimization procedure in the $t$-th iteration:
\begin{enumerate}
\item Reinitialize and optimize \(\phi\): At each iteration, reinitialize \(\phi\) randomly and optimize it to minimize \(\mathcal{L}_{attack} \) with \(\theta\) fixed, yielding
\begin{equation}
    \phi^{(t+1)} = \arg\min_{\phi} \mathcal L_{attack}(\phi, \theta^{(t)})
\end{equation}
\item Update \(\theta\): Fix \(\phi^{(t+1)}\) and perform one step of gradient descent on \(\theta\) to minimize \(\mathcal L_{learner}(\phi^{(t+1)}, \theta)\), yielding
\begin{equation}
   \theta^{(t+1)} = \theta^{(t)} - \eta \nabla_{\theta} \mathcal L_{learner}(\phi^{(t+1)}, \theta^{(t)})
\end{equation}
where \(\eta\) is the learning rate.
\end{enumerate}

Our proposed SITP offers several advantages in the privacy-preserving data publishing scenario:
\begin{itemize}[leftmargin=5.5mm]
    \item \textbf{Enhanced privacy preservation.} In the initial training phase, the surrogate attack model inevitably encodes the privacy information from the original graph, as it is trained to infer sensitive links on the basis of this information. If the attack model is continually trained without reinitialization, this privacy-related prior cannot be eliminated, leading to a significant divergence from the true attack model. Consequently, the graph learner may stray from developing a genuinely privacy-preserving graph.
    Reinitializing the attack model periodically in SITP, however, helps ensure that the learned graph achieves greater privacy preservation.

    \item \textbf{Stable convergence.} The surrogate attack model is optimized to converge in each update step, allowing the graph learner to converge more reliably. Theorem~\ref{thm:convergence} validates the convergence property of the graph learner $\mathcal{H}_\theta$ under our proposed SITP.

    \item \textbf{Efficient training process.} Compared with training the graph learner, training the surrogate attack model is relatively straightforward and converges quickly because of its fewer parameters and simpler training objectives. Therefore, retraining the attack model after each gradient descent step of the graph learner incurs tolerable computational overhead and can even reduce the overall training time compared with typical training protocols.

\end{itemize}

We also empirically compare SITP with typical training protocols in terms of privacy--utility trade-offs, convergence, and training efficiency, as shown in Sec.~\ref{sec:comparison_sitp}.

\subsection{Speed-up Strategies}
\label{sec:speedup}
The scalability and efficiency of the vanilla PPGSL may encounter limitations for two reasons:
\begin{enumerate}
\item[(1)] The space complexity of the graph learner $\mathcal H_\theta$, exacerbated by the FGP method, necessitates $\mathcal O(N^2)$ parameters for optimization, thereby posing a significant scalability issue when the graph is large.
\item[(2)] Each iteration for updating the learned graph $G'$ requires retraining the surrogate attack model until convergence to achieve accurate privacy leakage estimation, which is time-consuming. Thus, we propose two speed-up strategies to increase the scalability and efficiency of the PPGSL as follows.
\end{enumerate}

\textbf{PPGSL-sparse.}
Since graphs are typically sparse in practice (\textit{i.e.}, $|E| \ll N^2$), we can parameterize only the existing links and portions of nonexistent links in the original graph rather than parameterizing the whole adjacency matrix via the FGP method. Specifically, we define a parameter of the \textit{sampling factor} denoted by $k$, and we randomly sample a set of $(k \times |E|)$ nonexistent links. We then parameterize the edge weight of each existent and sampled nonexistent link.
By this strategy, we can substantially reduce the computational complexity of the PPGSL from $\mathcal O(N^2)$ to $\mathcal O((k+1) \times |E|)$. Moreover, a larger sampling size $k$ for nonexistent links leads to enhanced privacy protection (owing to the potential structural noise of adding edges) but with increased computational overhead.

\textbf{PPGSL-skip.}
Considering the subtle variation in the learned graph structure within each update, we can persistently utilize the surrogate attack model derived from the preceding graph structure, thereby reducing the update frequency of the surrogate attack model and reducing the overall model training time.
We define a parameter of \textit{update interval}, symbolized by $\mu$. Specifically, we retrain the surrogate attack model every $\mu$ iterations, as opposed to constant updates at each iteration. Amplifying $\mu$ may curtail time expenditures but yield less precise results.

The pseudocode of our proposed PPGSL framework with two speed-up strategies is shown in Algorithm~\ref{alg:dalia}.

\renewcommand{\algorithmicrequire}{\textbf{Input:}}  
\renewcommand{\algorithmicensure}{\textbf{Output:}} 

\begin{algorithm}[t]
    \small
  \caption{Pseudocode of the PPGSL Framework}
  \label{alg:dalia}
  \begin{algorithmic}[1]
    \Require
    the original graph $G=(A,X)$, the set of sensitive links $E_s$, hyper-parameters $\alpha$, $k$, $\mu$, the training epoch $N_1$ and learning rate $\eta_1$ of the surrogate attack model, the training epoch $N_2$ and learning rate $\eta_2$ of the privacy-preserving graph learner
    \Ensure the learned graph structure $A''$ for publishing
      \State Sample $(k\times|E|)$ nonexistent edges and construct the graph learner $\mathcal H_\theta$
      \State Initialize parameters of the graph learner $\mathcal H_\theta$
      \State Construct the surrogate attack model $f_{\phi}$, $\delta_\varphi$
      \For{$e_2 = 0$; $e_2 < N_2$; $e_2 ++$ }
      \State Generate a graph structure, $A'$ $\leftarrow$ $\mathcal H_\theta (G)$
      \If {$e_2\, \% \,\mu=0$}
      \State Initialize parameters of surrogate attack model $f_{\phi}$, $\delta_\varphi$
      \For{$e_1 = 0$; $e_1<N_1$; $e_1 ++$ }
          \State Sample sets of existent and nonexistent links from $A'$
          \State Calculate $\mathcal L_{attack}$ in Eq.~\ref{eq:gnn_encoder} based on $A'$
          \State Update parameters, $\phi$ $\leftarrow$ $\phi-\eta_1\cdot\frac{\partial \mathcal L_{attack}}{\partial \phi}$, $\varphi$ $\leftarrow$ $\varphi-\eta_1\cdot\frac{\partial \mathcal L_{attack}}{\partial \varphi}$
      \EndFor
      \EndIf
      \State Calculate $\mathcal L_{learner}$ in Eq.~\ref{eq:overall_loss} based on $f_{\phi}$, $\delta_\varphi$, $A'$, and $A$
      \State Update parameters, $\theta$ $\leftarrow$ $\theta-\eta_2\cdot\frac{\partial \mathcal L_{learner}}{\partial \theta}$
      \EndFor
     \State Generate a graph structure, $A'$ $\leftarrow$ $\mathcal H_\theta (G)$
     \State Discretize the graph structure using Bernoulli sampling function $\mathcal T$, $A''$ $\leftarrow$ $\mathcal{T}(A')$
     \State \Return $A''$
  \end{algorithmic}
\end{algorithm}

\section{Theoretical Analyses}
\label{sec:theoretical}

Rigorous theoretical analyses prove the convergence property of the PPGSL's training protocol, the optimality of its privacy--utility trade-off, and its generalizability across various attack methods.
The detailed proofs are provided in Appendix~\ref{appendix:proof}.

\begin{theorem}[\textbf{Convergence of the PPGSL}]
\label{thm:convergence}
In PPGSL training procedures, graph learner $\mathcal{H}_\theta$ converges under SITP.
\end{theorem}

\begin{proof}[Proof Sketch]
Under the training of the PPGSL with SITP, the following inequality holds from the $t$-th iteration to the $(t+1)$-th iteration (the detailed proof is in Proposition~\ref{co:convergence1}):
\begin{equation}
\mathbb{E}[\mathcal{L}_{learner}(\phi^{(t+1)}, \theta^{(t+1)})] \leq\mathbb{E}[\mathcal{L}_{learner}(\phi^{(t)}, \theta^{(t)})]
\end{equation}
which implies that each full iteration (reinitializing and optimizing \(\phi\), then updating \(\theta\)) results in a nonincreasing expected value of the loss function \(\mathcal{L}_{learner}\).
Since the expected value of \(\mathcal{L}_{learner}\) is nonincreasing at each iteration and is assumed to be lower-bounded, by the Monotone Convergence Theorem, \(\mathbb{E}[\mathcal{L}_{learner}(\phi, \theta)]\) converges to a stable value as the iteration index $t$ increases.
Therefore, graph learner $\mathcal{H}_\theta$ converges.
\end{proof}


\begin{theorem}[\textbf{Empirical Optimal Privacy--utility Trade-off of the PPGSL}]
Minimizing the training objective \(\mathcal{L}_{learner}\) of the PPGSL achieves the empirical optimal privacy--utility trade-off at a specified utility level.
\end{theorem}

\begin{proof}[Proof Sketch]
Relying on the Lagrangian dual method, we can prove that minimizing \(\mathcal{L}_{learner}\) can be reinterpreted as minimizing privacy loss while adhering to a given utility constraint (the detailed proof is in Proposition~\ref{theorem:lagrange_dual}). This ensures that the empirical optimal privacy--utility trade-off is attained.
\end{proof}

\begin{theorem}[\textbf{Generalized Privacy Protection Performance of the PPGSL}]
The PPGSL provides a lower bound of the privacy protection level on sensitive links in a graph regardless of the sensitive link inference model adopted by attackers.
The inference error probability $p(E_s\neq \mathcal M^*(G'))$ of any inference model $\mathcal M^*$ that attempts to infer \( E_s \) from the published graph \( G' \) is lower-bounded by:
\begin{equation}
p(E_s\neq \mathcal M^*(G'))\ge \frac{H(E_s)-I(G';E_s)-1}{\log|\mathcal{E}_s|}
\end{equation}
where $|\mathcal{E}_s|$ denotes the cardinality of the set of possible values of \( E_s \), $I(\cdot,\cdot)$ represents mutual information, and $H(\cdot)$ denotes information entropy. Furthermore, this lower bound increases during the PPGSL training process.
\end{theorem}

\begin{proof}[Proof Sketch]
Motivated by prior studies that confirmed a relationship between the inference success of any algorithm and mutual information measures~\cite{du2012privacy,han2023hyobscure}, we establish a connection between PPGSL and mutual information. Let \( I(\cdot, \cdot) \) denote mutual information; then, the privacy goal of our problem can be reformulated as $\min_{\theta} I(G'; E_s)$.
We aim for the learned graph \( G' \) to contain as little mutual information about sensitive links as possible. In the PPGSL training process, the mutual information $I(G';E_s)$ decreases (the detailed proof is in Proposition~\ref{co:fano}).

According to Fano's inequality, the inference error probability $p(E_s\neq \mathcal M^*(G'))$ of any inference model $\mathcal M^*$ that attempts to infer \( E_s \) from \( G' \) is lower-bounded by
$\frac{H(E_s)-I(G';E_s)-1}{\log|\mathcal{E}_s|}$, where $H(E_s)$ and $\log|\mathcal{E}_s|$ are constants for a given private information variable $E_s$.
Consequently, as the mutual information \( I(G'; E_s) \) decreases, the lower bound on the inference error probability \( p(E_s \neq \mathcal{M}^*(G')) \) increases, where $\mathcal M^*$ is any sensitive link inference model trained on the published graph $G'$. Therefore, the PPGSL inherently increases the inference error probability by maximizing this lower bound, regardless of the inference algorithm used.
\end{proof}

\begin{table}[t]
\footnotesize
  \caption{Statistics of the datasets.}
  \label{tbl:datasetstats}
  \centering
  \begin{tabular}{lcccc}
    \toprule
    \textbf{Dataset} & \textbf{\#Nodes} & \textbf{\#Links} & \textbf{\#Features} & \textbf{\#Labels}  \\
    \midrule
    PolBlogs & 1,490 & 19,025 & 0 & 2 \\
    LastFMAsia & 7,624 & 27,806 & 128 & 18\\
    DeezerEurope & 28,281 & 185,504 & 128 & 2 \\
    Cora & 2,708 & 10,556 & 1,433 & 7 \\
    CiteSeer & 3,327 & 9,104 & 3,703 & 6 \\
    PubMed & 19,717 & 88,648 & 500 & 3 \\
    \bottomrule
  \end{tabular}
\vspace{-1em}
\end{table}

\section{Experiments}
\subsection{Experimental Setup}
\subsubsection{Datasets}
\label{sec:datasets}
We conduct experiments on six commonly used real-life datasets\footnote{All datasets are available from the PyTorch Geometric libraries: \href{https://pytorch-geometric.readthedocs.io/en/latest/modules/datasets.html}{https://pytorch-geometric.readthedocs.io/en/latest/modules/datasets.html}.}, including three social networks: \textit{PolBlogs}, \textit{LastFMAsia} and \textit{DeezerEurope}; and three citation graphs: \textit{Cora}, \textit{CiteSeer} and \textit{PubMed}.
Their statistics are summarized in Table~\ref{tbl:datasetstats}.
The dataset descriptions are as follows:
\begin{itemize}[leftmargin=5.5mm]
    \item \textbf{PolBlogs}~\cite{adamic2005political} is a relationship network of political blogs, where nodes represent blogs and edges signify links extracted from their front pages. Each node is labeled as either liberal or conservative.
    \item \textbf{LastFMAsia}~\cite{rozemberczki2020characteristic} is a social network of users from Asian countries on the music service LastFM. Nodes represent users, edges denote friendships, and node labels are users' home countries, with features based on preferred artists.
    \item \textbf{DeezerEurope}~\cite{rozemberczki2020characteristic} is a social network of European Deezer users. Nodes represent users, and links represent mutual follower relationships. Node labels stand for gender, and features are derived from favorite artists.
    \item \textbf{Cora}, \textbf{Citeseer} and \textbf{PubMed}~\cite{yang2016revisiting} are citation networks, where nodes stand for documents and edges denote citations. Each node has a bag-of-words feature vector and is labeled by the document category.
\end{itemize}

Following prior works~\cite{yu2019target,han2023privacy}, we arbitrarily mask 10\% of existing links within each graph as sensitive links and randomly select an equivalent quantity of nonexistent links as negative samples for evaluation.

\subsubsection{Sensitive Link Inference Attacks and Privacy Metric}
\label{sec:inference_attack}

\begin{table}[t]
\footnotesize
  \caption{Traditional node proximity metrics for link inference attacks. $\mathcal N(x)$: the neighbor set of vertex $x$; \textit{Order}: the maximum hop of neighbors needed to calculate the proximity score.}
  \label{tbl:proximity_metrics}
  \centering
  \begin{tabular}{lcc}
    \toprule
    \textbf{Metrics} & \textbf{Formula} & \textbf{Order}  \\
    \midrule
    Common Neighbor (CN) & $| \mathcal N (x) \cap \mathcal N (y) |$ & first \\
    Adamic-Adar (AA) & $ \sum_{z \in {\mathcal N (x) \cap \mathcal N (y)}} \frac{1}{\log {|\mathcal N (z)|}}$ & second \\
    Resource Allocation (RA) & $\sum_{z \in {\mathcal N (x) \cap \mathcal N (y)}} \frac{1}{|\mathcal N (z)|}$ & second \\
    \bottomrule
  \end{tabular}
\vspace{-1em}
\end{table}

To assess the privacy protection efficacy of the PPGSL, we employ diverse methods to conduct various sensitive link inference attacks on its generated published graph. A lower inference success rate indicates stronger privacy protection. We consider eight distinct attacks that fall into the following two categories:
\begin{itemize}[leftmargin=5.5mm]
    \item \textbf{Structure-based attacks:} Common Neighbor (CN), Adamic-Adar (AA), Resource Allocation (RA), and SEAL~\cite{zhang2018link}.
    These methods calculate node proximity scores as probabilities for the existence of sensitive links. The three traditional metrics are presented in Table~\ref{tbl:proximity_metrics}. SEAL extracts local subgraphs around each target link and trains a link prediction model based on them. We implement SEAL with node attributes disregarded and the hop number set to 2.

    \item \textbf{Embedding-based attacks:} node2vec~\cite{grover2016node2vec} with cosine similarity (N2V+sim), GAE~\cite{kipf2016variational} with cosine similarity (GAE+sim), and combinations with LinearSVC (N2V+ML, GAE+ML). These methods first obtain node embeddings through graph representation learning and then compute the probability of sensitive link existence via similarity metrics (\textit{e.g.}, cosine similarity~\cite{he2021stealing,han2023privacy}) or classifiers (\textit{e.g.}, LinearSVC~\cite{sun2019infograph,you2020graph}). For classifiers, the input features are derived from concatenating the embeddings of two nodes, using published graph links as positive samples and an equal number of unconnected node pairs as negative samples for training.
\end{itemize}

In alignment with previous works~\cite{wu2022linkteller,han2023privacy}, we adopt the AUC (area under the ROC curve) as the privacy metric.
A higher AUC indicates that it is easier for an attacker to infer sensitive links, leading to higher privacy risks and worse privacy protection effect.

\subsubsection{Utility Evaluation Tasks and Metrics}
We leverage two widely used graph-based tasks, \emph{i.e.,} link prediction and node classification, to evaluate the utility of the published graph.

\textbf{Link prediction} attempts to predict the missing or potential links on a graph~\cite{lu2011link,zhang2018link}. We randomly mask 10\% of the real links from the original graph as positive testing links and sample the same number of nonexistent links as negative testing links. The set of testing links is nonoverlapping with the set of sensitive links.
For a published graph, we train a GNN model via unsupervised loss to procure node embeddings~\cite{hamilton2017inductive} and then compute embedding similarity to predict the existence of testing links.
We employ the AUC as the performance metric of link prediction.

\textbf{Node classification} aims to correctly classify the unlabeled nodes on the basis of the proportion of labeled nodes~\cite{he2021stealing, wu2022linkteller}. We randomly split the nodes of the original graph, allotting 30\% of the nodes for training and reserving the remaining 70\% for testing. Given a published graph, we train a two-layer GCN model~\cite{kipf2016semi} in a semisupervised manner~\cite{kipf2016semi} to predict node labels.
The F1 score is used to measure the performance of node classification.

\subsubsection{Baselines} We compare the privacy--utility trade-off performance of PPGSL with that of seven baselines: 

\textbf{Random}~\cite{yu2019target} randomly removes partial links and adds the same number of new links to generate perturbed graphs.

\textbf{DICE}~\cite{zugner_adversarial_2019} generates perturbed graphs by deleting links connected to nodes with sensitive links and adding links between nodes without sensitive links.

\textbf{PrivGraph}~\cite{YZDCCS23} exploits community information to generate synthetic graphs with differential privacy guarantees. 

\textbf{EdgeRand}~\cite{wu2022linkteller} is an edge--DP defense strategy, which randomly flips each entry in the adjacency matrix according to a Bernoulli random variable.

\textbf{LapGraph}~\cite{wu2022linkteller} also guarantees edge--DP. It precalculates the original graph density using a small privacy budget and uses that density to clip the perturbed adjacency matrix.

\textbf{RW-LP}~\cite{mittalPS13} 
replaces real links with fake links between the starting and terminal nodes of random walks.

\textbf{PPNE}~\cite{han2023privacy} samples some candidate node pairs for link perturbation and then iteratively selects the optimal pairs with a high privacy--utility trade-off to yield the perturbed graph.

\subsubsection{Running Environment}
Our experimental platform is a PC with an Intel i5 12600KF CPU (10 cores @ 3.7 GHz), 32 GB of RAM, and an NVIDIA RTX 4070Ti SUPER GPU (16 GB).
The operation system is Ubuntu 22.04 LTS.
We utilize Python 3.11, PyTorch 1.12.1, and PyTorch Geometric 2.3.1.

\subsubsection{PPGSL Implementation}
If not specified, we employ the PPGSL-skip ($\mu=50$) and the PPGSL-sparse ($k=1$) speed-up strategies. We
adjust $\alpha\in [0,0.01]$ to obtain graphs with varying levels of privacy preservation\footnote{Guidance for selecting $\alpha$ is provided in Appendix~\ref{sec:alpha_Selection}.}.
For the GNN encoder $f_{\phi}$, we utilize a 2-layer GCN with hidden dimensions of [128, 64] and set the training epoch at 500. For the prediction head $\delta_\varphi$, we choose cosine similarity\footnote{To demonstrate the robustness of our approach, we also evaluate its performance with various surrogate attack models. The results and analysis are presented in Appendix~\ref{sec:surrogate_robust}.}.
For the graph learner $\mathcal H_\theta$, we fix the training epoch at 500 and choose the Adam optimizer with a learning rate of 0.5 to optimize the parameterized graph structure.
We run all the experiments five times and report the average results.

\subsection{Experimental Results}

\subsubsection{Privacy--utility Trade-off Performance Compared with Baselines}
\label{sub:exp_tradeoff}

\begin{figure*}[t]
    \centering
    \includegraphics[width=\linewidth]{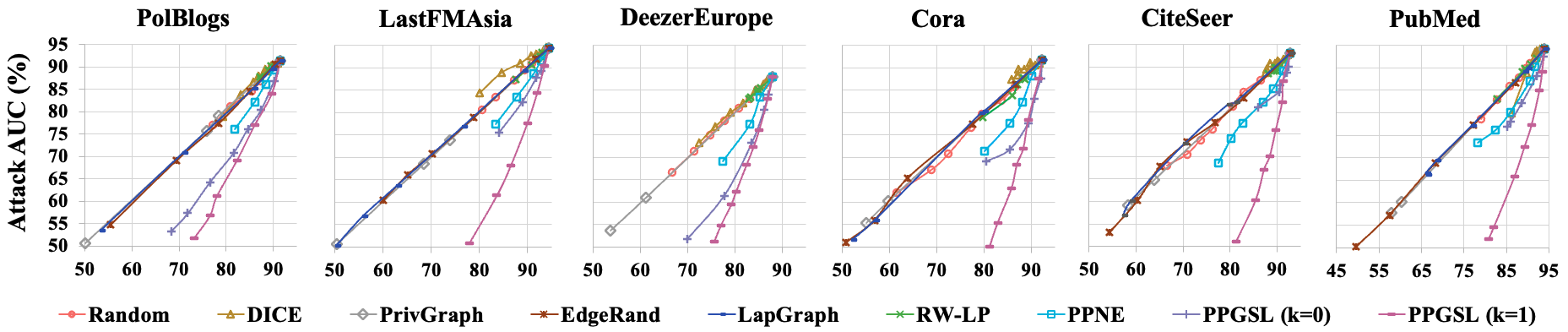}
    \vspace{-1em}
    \caption{Privacy--utility trade-off performance of the PPGSL. X-axis: AUC (\%) of the utility task in terms of link prediction; Y-axis: AUC (\%) of the GAE+sim attack method. Points at the top right of the curves represent evaluations of the original graph.
    The EdgeRand and LapGraph methods result in an out-of-memory error on \textit{DeezerEurope}.}
    \label{fig:result_tradeoff_link}
\end{figure*}

\begin{figure*}[t]
    \centering
    \includegraphics[width=\linewidth]{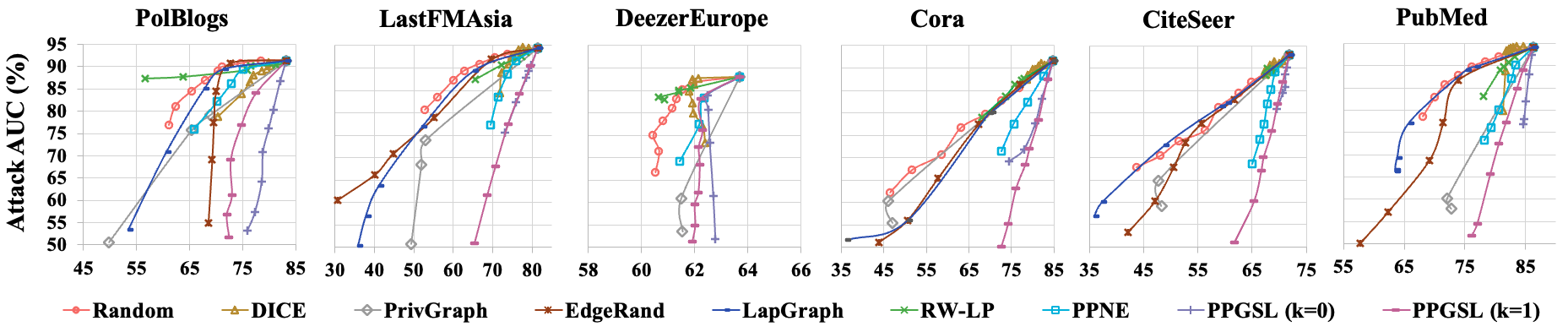}
    \vspace{-1em}
    \caption{Privacy--utility trade-off performance of the PPGSL. X-axis: F1 score (\%) of the utility task in terms of node classification; Y-axis: AUC (\%) of the GAE+sim attack method. Points at the top right of the curves indicate evaluations of the original graph.
    The EdgeRand and LapGraph methods result in an out-of-memory error on \textit{DeezerEurope}.}
    \label{fig:result_tradeoff_node}
    \vspace{-1em}
\end{figure*}

We report the privacy--utility trade-off performance of the PPGSL and five baselines under the utility tasks of link prediction (Fig.~\ref{fig:result_tradeoff_link}) and node classification (Fig.~\ref{fig:result_tradeoff_node}).
In this part, we employ GAE+sim as the attack method because of its widespread use and superior attack performance~\cite{han2023privacy,he2021stealing}.
We plot two variants of the PPGSL where $k=0$ or $k=1$: the PPGSL ($k=0$) signifies the parameterization solely of existing links, confining the structure perturbation to edge deletion; the PPGSL ($k=1$) also parameterizes a subset of nonexistent links, expanding the structure perturbation to encompass both edge deletion and addition.

Herein, a low value on the y-axis means higher privacy protection (lower attack performance), whereas a high value on the x-axis indicates better utility. We see that directly releasing original graphs poses a significant risk of privacy leakage for sensitive links. The points at the top right of the curves denote the attack AUC and utility evaluations on the original graphs, where these attacks achieve extremely high attack AUC scores.

Note that by modifying the hyperparameters to control the privacy--utility trade-off, we can draw a line for each method to show how utility changes with varying levels of privacy protection.
Most baselines exhibit limited privacy protection capability, with the lowest attack AUC failing to reach 50\% (almost equivalent to random guessing, representing optimal privacy protection). In contrast, the PPGSL ($k=1$) easily achieves an attack AUC of approximately 50\%, demonstrating strong privacy-safeguarding ability.

With respect to the performance of the privacy--utility trade-off, the PPGSL and the baselines show the same pattern: with greater perturbation, the privacy attack becomes more difficult, and the utility decreases. More importantly, we observe that the PPGSL ($k=0$) and the PPGSL ($k=1$) usually provide greater privacy protection given the same utility, thereby achieving a better privacy--utility trade-off (\textit{i.e.}, the PPGSL often appears in the lower right corner of the figures). Specifically, for \textit{Cora} in Fig.~\ref{fig:result_tradeoff_link}, when the utility level (AUC of link prediction) is approximately 81\%, the attack AUC of the PPGSL ($k=1$) can be reduced to approximately 50\%.

Moreover, we note that the PPGSL ($k=0$) and the PPGSL ($k=1$) exhibit divergent privacy--utility trade-off performances under different utility tasks. For link prediction (Fig.~\ref{fig:result_tradeoff_link}), the PPGSL ($k=1$) presents the best privacy--utility trade-off across all six datasets; for node classification (Fig.~\ref{fig:result_tradeoff_node}), however, the PPGSL ($k=0$) attains the best privacy--utility trade-off on more datasets than the PPGSL ($k=1$) does. This disparity could be attributed to the fact that edge addition introduces greater disruption in node classification; thus, the PPGSL ($k=1$) incurs greater utility loss when safeguarding privacy, whereas link prediction tasks are less susceptible to edge addition.

\subsubsection{Privacy Protection Effects Against Various Inference Attacks}
\label{sub:exp_privacy}

We conduct experiments to evaluate the PPGSL against eight sensitive link inference attacks, aiming to verify its generalized privacy protection effectiveness.
Figs.~\ref{fig:result_attack_link} and~\ref{fig:result_attack_node} show the results on six datasets under the utility tasks of link prediction and node classification, respectively.
Our results show that the PPGSL effectively defends against various inference attacks, as the PPGSL greatly decreases the attack AUC with tolerable utility loss on downstream tasks. Specifically, the PPGSL reduces the AUC of most attack methods to approximately 50\% on \textit{Cora} while preserving strong performances in link prediction (AUC of 81\%) and node classification (F1 score of 73\%).
In addition, the fluctuation in node classification results on \textit{DeezerEurope} is relatively minor. This can be attributed to the limited influence of graph structure information on this task for \textit{DeezerEurope}, as an MLP model using solely node features achieves a sufficiently high F1 score of 63\%.

\begin{figure*}[t]
    \centering
    \includegraphics[width=\linewidth]{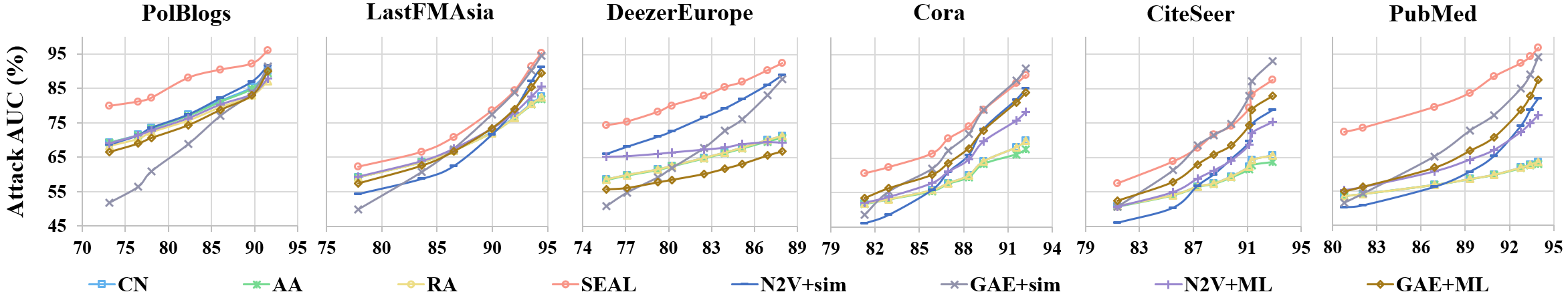}
    \vspace{-1em}
    \caption{Privacy protection effect of PPGSL against various attack methods under the link prediction utility task. X-axis: AUC (\%) of link prediction.}
    \label{fig:result_attack_link}
\end{figure*}

\begin{figure*}[t]
    \centering
    \includegraphics[width=\linewidth]{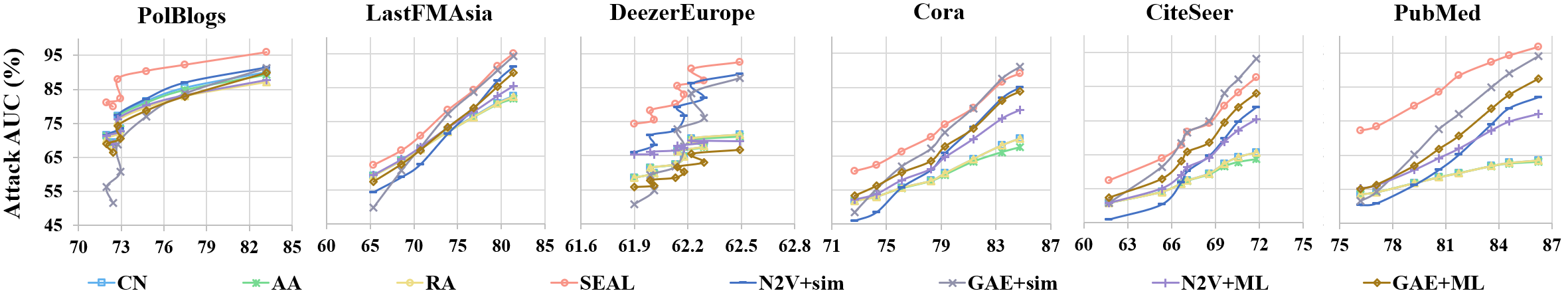}
    \vspace{-1em}
    \caption{Privacy protection effect of PPGSL against various attack methods under the node classification utility task. X-axis: F1 score (\%) of node classification.}
    \label{fig:result_attack_node}
    \vspace{-1em}
\end{figure*}

\subsubsection{Parameter Sensitivity}

\begin{figure}[t]
    \centering
    \includegraphics[width=\linewidth]{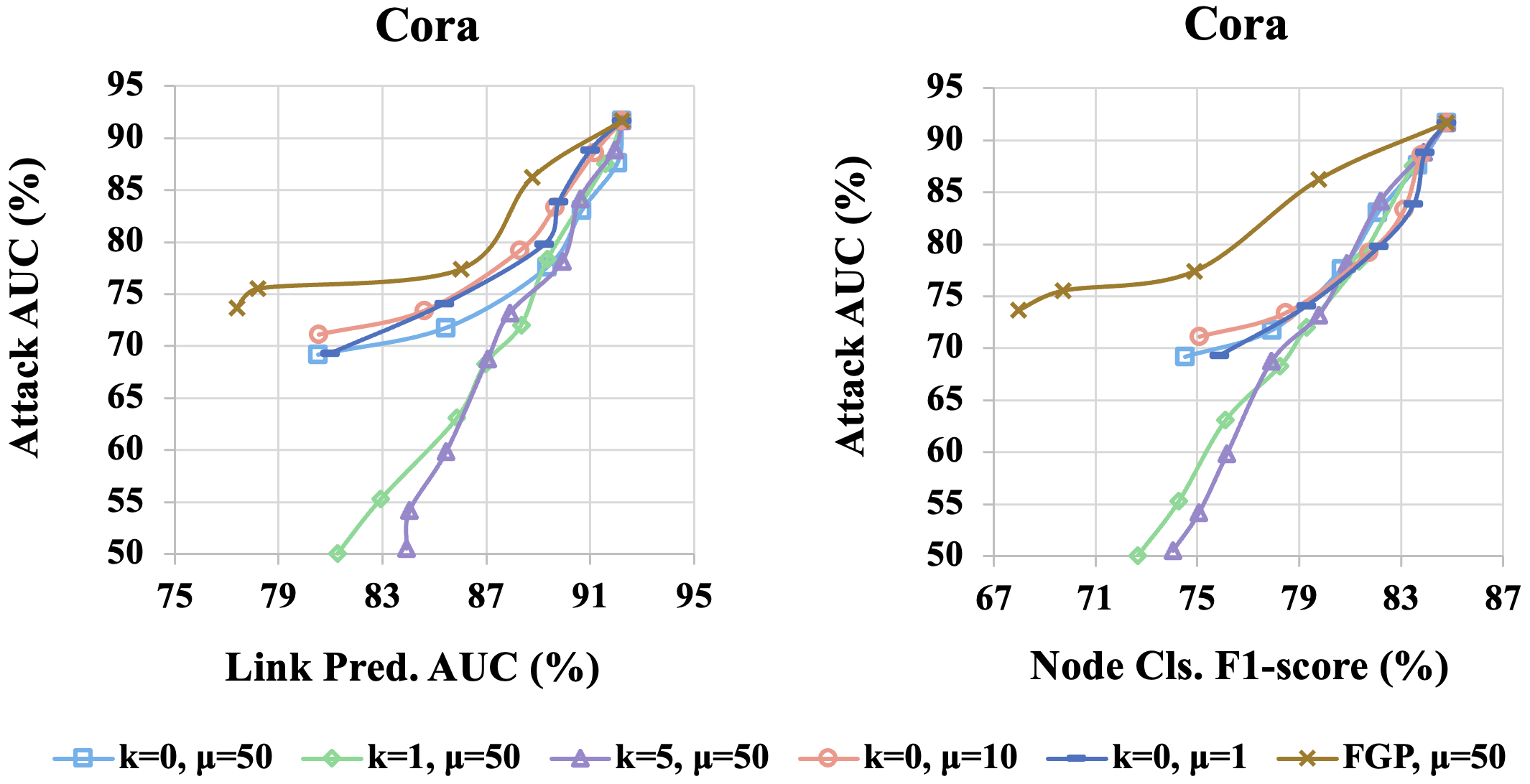}
    \vspace{-1em}
    \caption{Parameter analysis of $k$ and $\mu$ on \textit{Cora}. Left: link prediction utility task; right: node classification utility task.}
    \label{fig:result_parameter_k_mu}
\end{figure}

\begin{figure}[t]
    \centering
    \includegraphics[width=\linewidth]{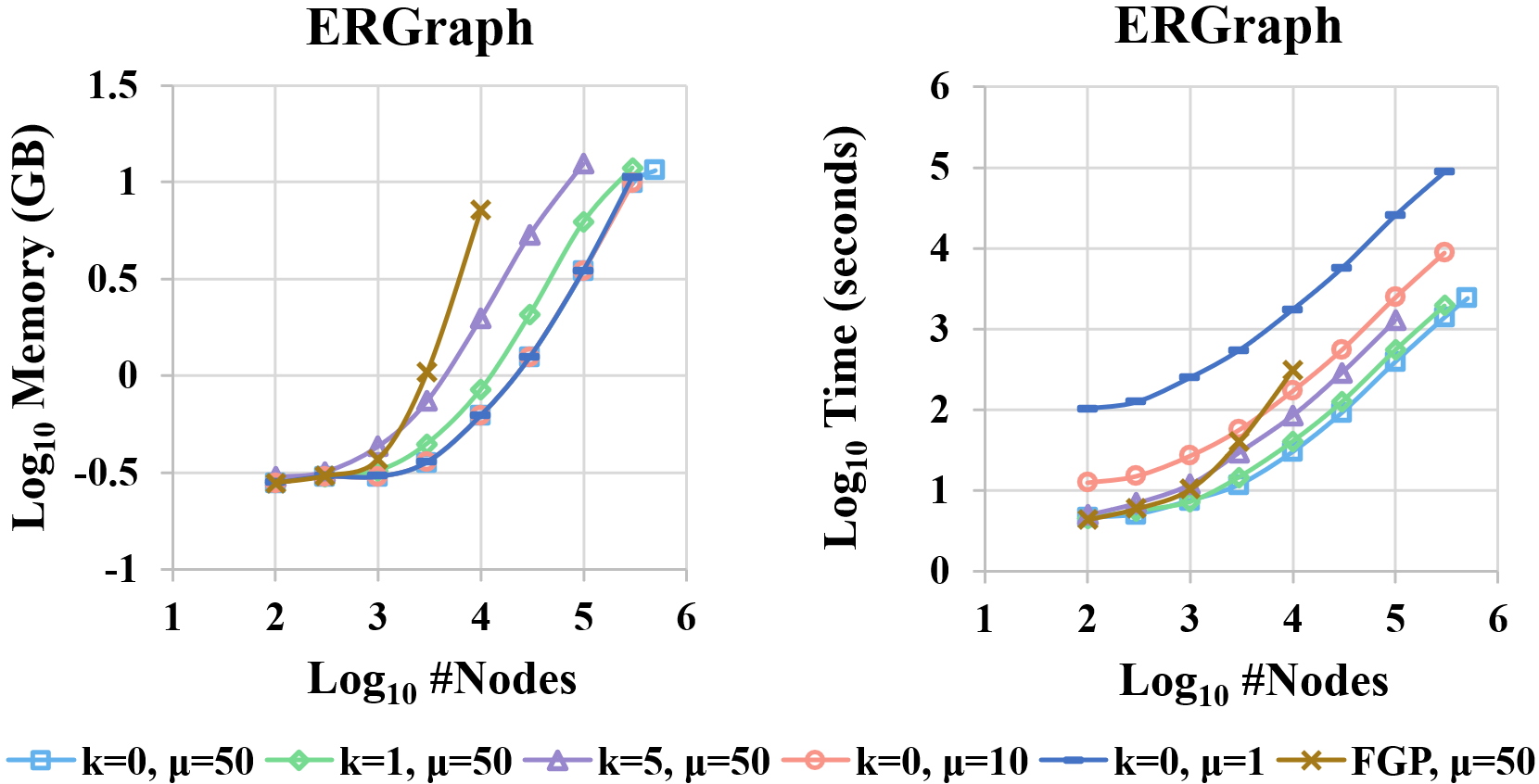}
    \vspace{-1em}
    \caption{Memory usage (left) and time consumption (right) for training the PPGSL on Erdos--Renyi graphs with varying node quantities. The training epoch is set to 500, as the PPGSL typically converges within 500 iterations.}
    \label{fig:result_scala}
    \vspace{-1em}
\end{figure}

Sec.~\ref{sec:speedup} introduces the PPGSL-sparse with a sampling factor $k$ and the PPGSL-skip with an update interval $\mu$ to scale up the PPGSL.
Here, we train the PPGSL with varying $k$ and $\mu$, and the results on \textit{Cora} are shown in Fig.~\ref{fig:result_parameter_k_mu}. The PPGSL-FGP uses the PPGSL's model architecture that parameterizes the full adjacency matrix.
We observe that $k=50$ and $\mu=50/10/1$ yield similar trade-off performances, indicating that $\mu$ has little effect on the performance of the PPGSL. Both the PPGSL ($k=1$) and the PPGSL ($k=5$) demonstrate similar performance, significantly outperforming the PPGSL ($k=0$) and the PPGSL-FGP. This is likely because the PPGSL ($k=0$) removes only edges without adding any, resulting in minimal perturbation and a lower maximum level of privacy protection. In contrast, the PPGSL-FGP imposes no restrictions on the addition of edges, leading to greater perturbation and consequently greater utility loss. The PPGSL ($k=1$) and the PPGSL ($k=5$), however, limit the maximum number of added edges, resulting in moderate perturbation. They achieve a significant level of privacy protection while limiting utility loss, thus offering better privacy--utility trade-off performance. Hence, we recommend setting $k=1$ and $\mu=50$ in most cases (for node classification tasks, $k=0$ can also be tried as discussed in Sec.~\ref{sub:exp_tradeoff}).

\subsubsection{Scalability of the PPGSL}

We present the memory usage and training time of the PPGSL with simulated networks of diverse scales in Fig.~\ref{fig:result_scala}.
Specifically, we generate a series of Erdos--Renyi graphs with node counts ranging from 100 to 500,000
and an average node degree of 10.
In general, the memory usage and training time of the PPGSL ($k=0, \mu=50$) increase approximately linearly with the node count. With $\mu = 50$, we compare the PPGSL-sparse with $k\in\{0,1,5\}$ and the PPGSL-FGP. The PPGSL-sparse drastically reduces memory usage, particularly when $k$ is small.
With $k=0$, we compare the PPGSL-skip with $\mu\in\{1,10,50\}$, and it is apparent that as $\mu$ increases, the time consumption diminishes remarkably, whereas the privacy--utility trade-off performance remains almost unchanged (please see Fig.~\ref{fig:result_parameter_k_mu}).

Notably, for a graph with up to 100,000 nodes, the PPGSL can complete training with our recommended settings (\(k=1\), \(\mu=50\)) in just 9.2 minutes, with an average memory usage of 6.2 GB.
In contrast, for a graph of the same size, the most comparable method, PPNE, requires approximately 16.7 minutes per iteration and converges around the 3,000th iteration~\cite{han2023privacy}. Consequently, the PPGSL achieves a training speedup of \(16.7 \times 3,000 \div 9.2 \approx 5,445.7\) times faster than the PPNE for a complete training process.

\subsubsection{Convergence of the PPGSL}
\label{sec:convergence}

\begin{figure}[t]
\centering
\subfloat[\footnotesize{Training loss of $\mu=1$.}\label{fig:conver_1_loss}]
{
\includegraphics[width=0.48\linewidth]{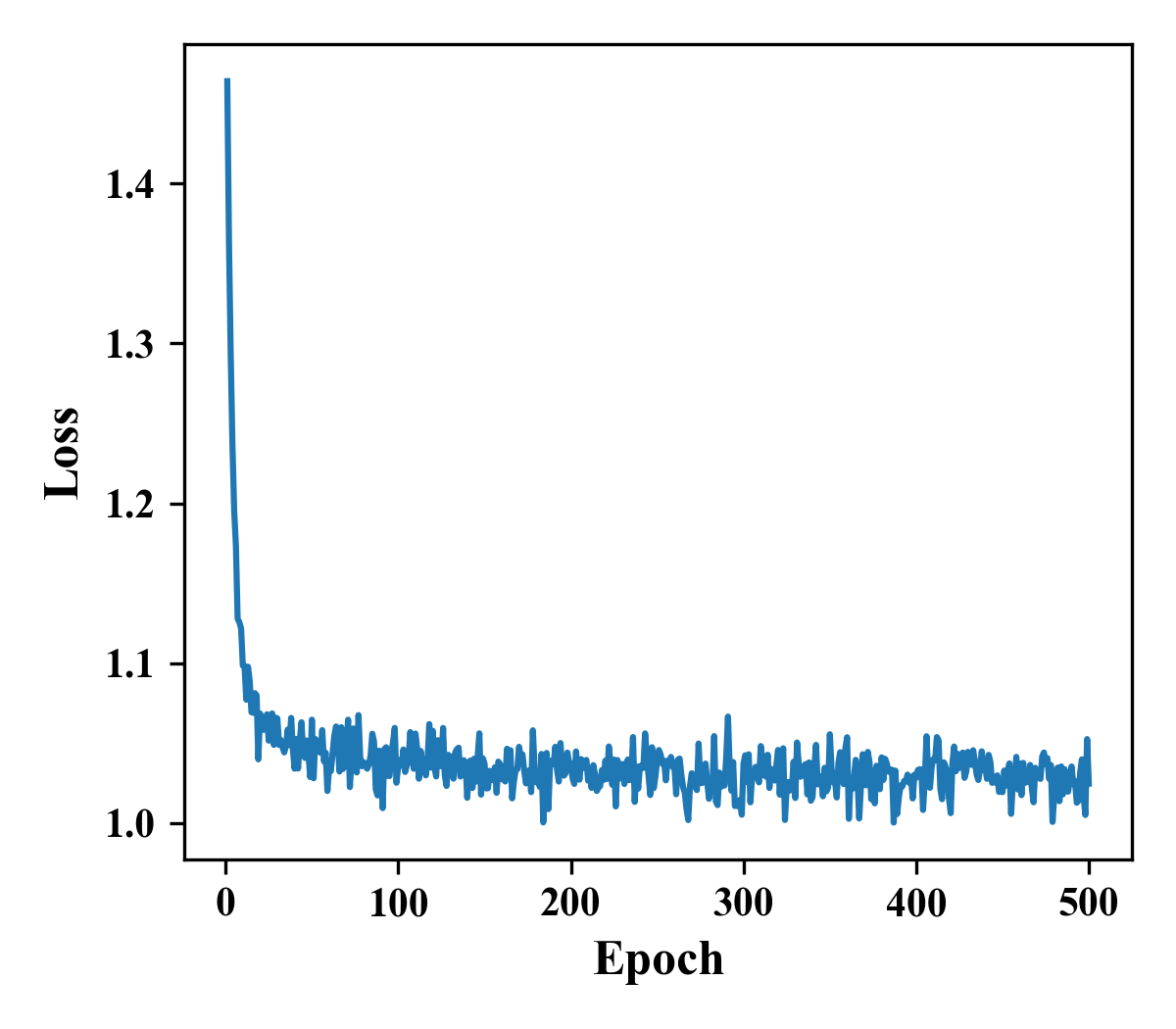}
}
\subfloat[\footnotesize{Evaluation results of $\mu=1$.}\label{fig:conver_1_eva}]
{
\includegraphics[width=0.48\linewidth]{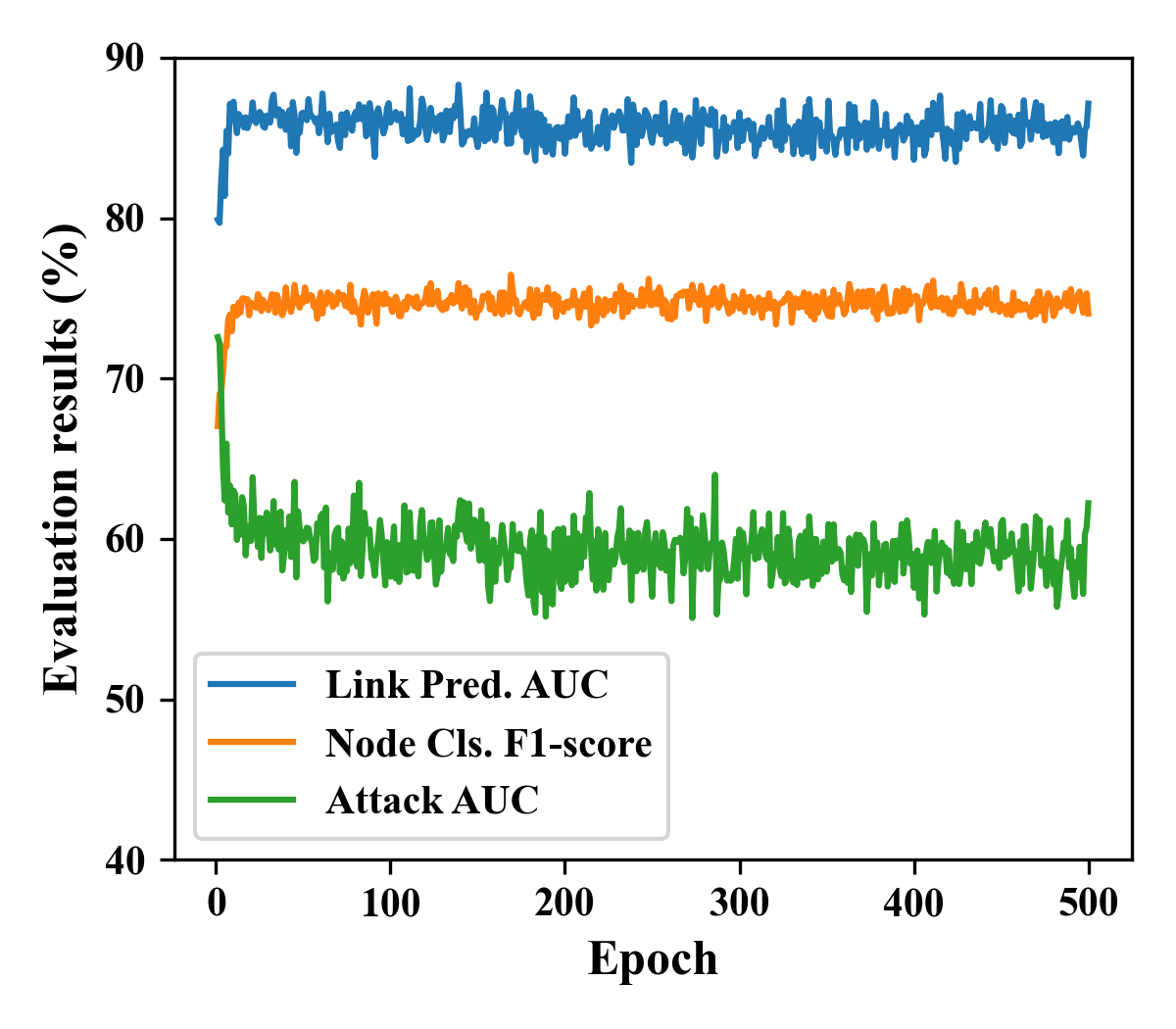}
}
\vspace{-0.8em}

\subfloat[\footnotesize{Training loss of $\mu=10$.}\label{fig:conver_10_loss}]
{
\includegraphics[width=0.48\linewidth]{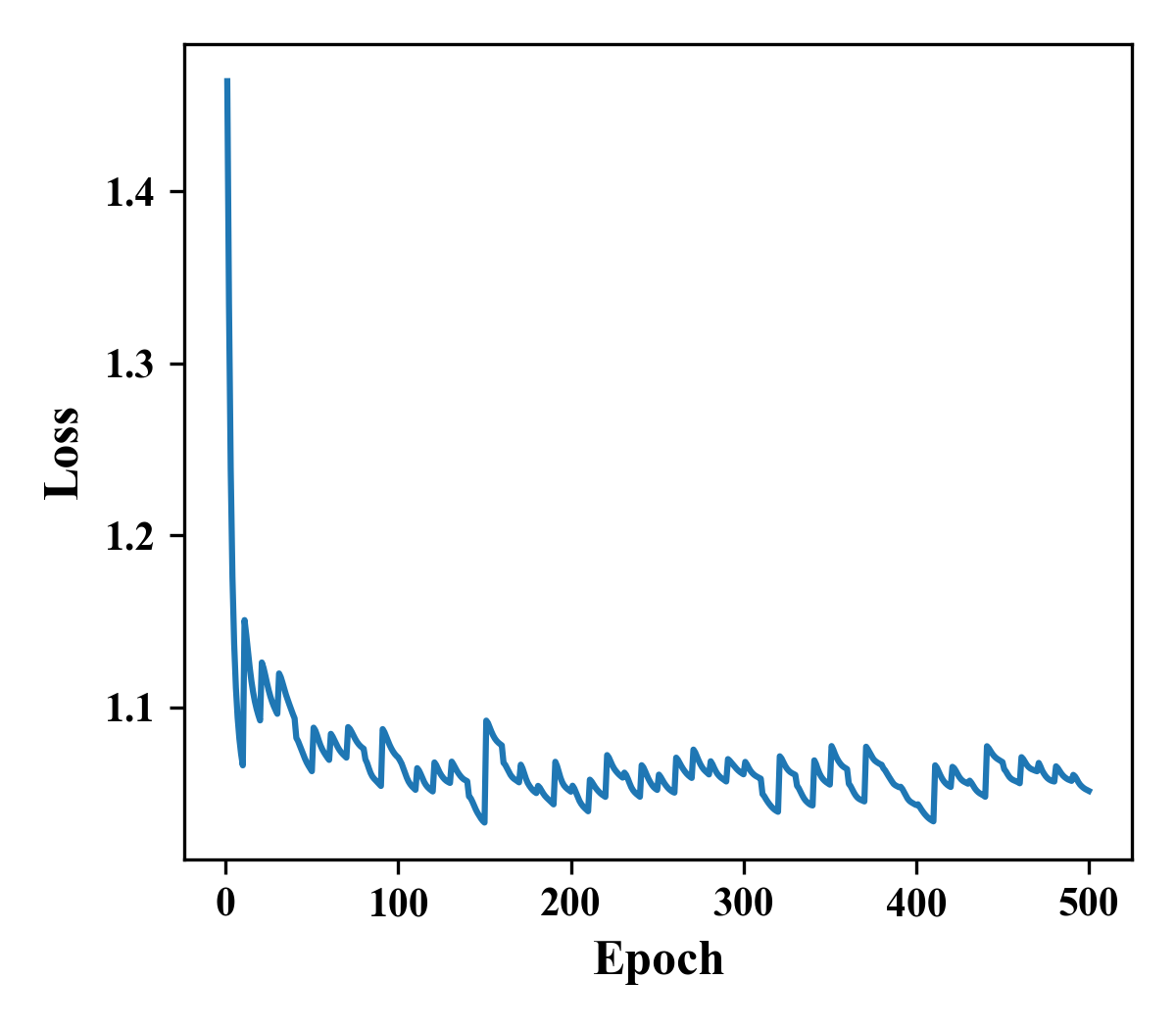}
}
\subfloat[\footnotesize{Evaluation results of $\mu=10$.}\label{fig:conver_10_eva}]
{
\includegraphics[width=0.48\linewidth]{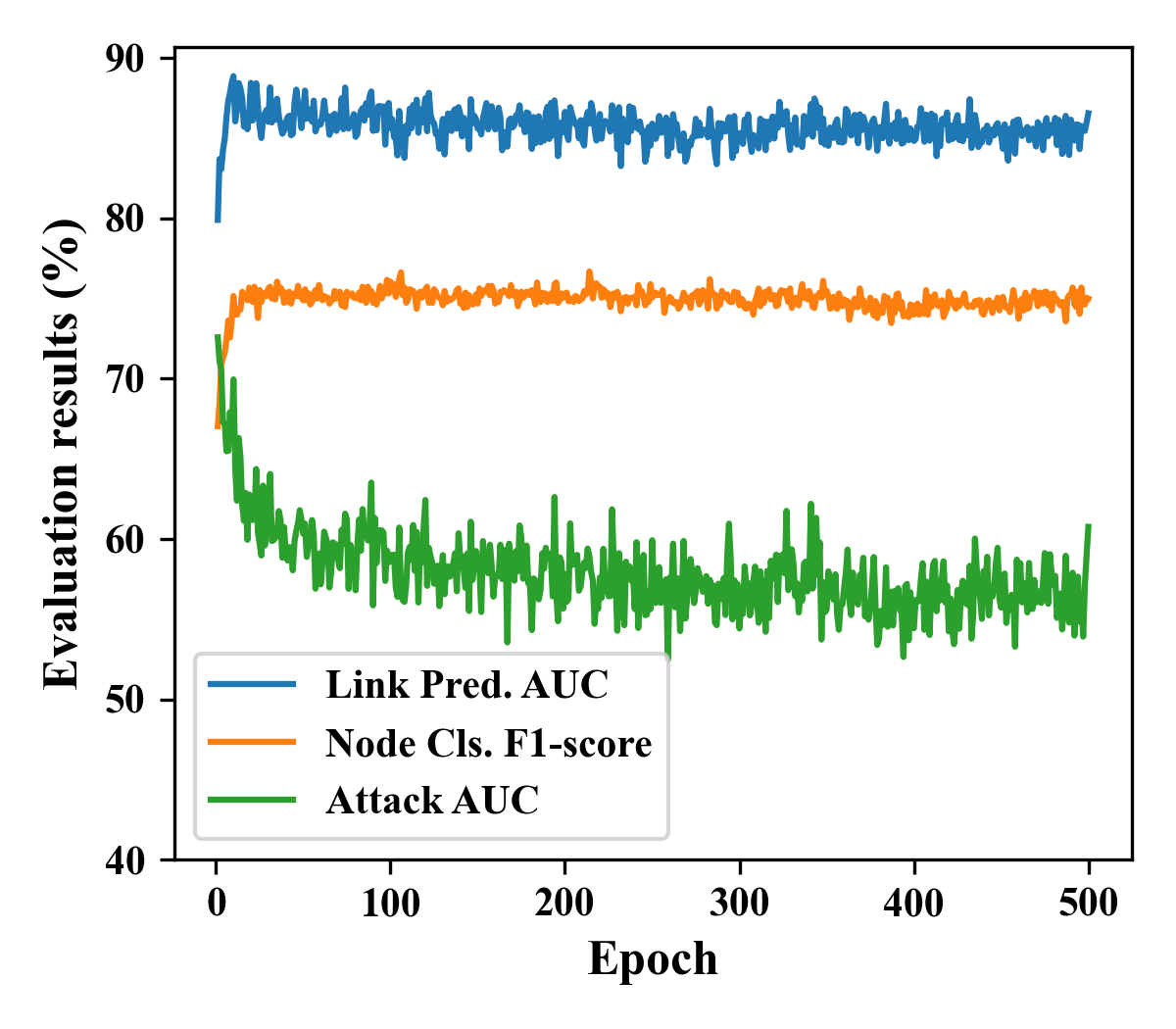}
}
\vspace{-0.8em}

\subfloat[\footnotesize{Training loss of $\mu=50$.}\label{fig:conver_50_loss}]
{
\includegraphics[width=0.48\linewidth]{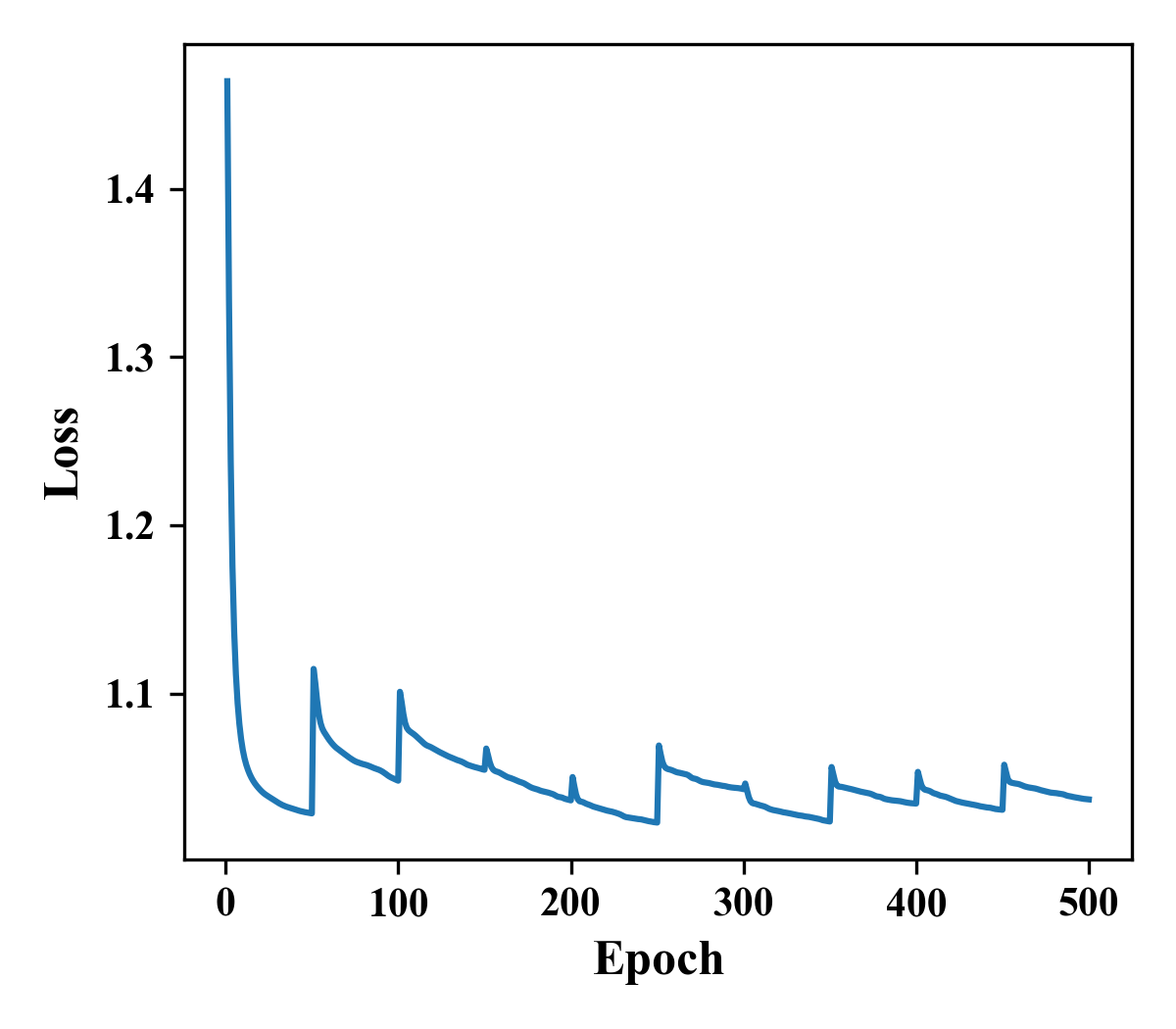}
}
\subfloat[\footnotesize{Evaluation results of $\mu=50$.}\label{fig:conver_50_eva}]
{
\includegraphics[width=0.48\linewidth]{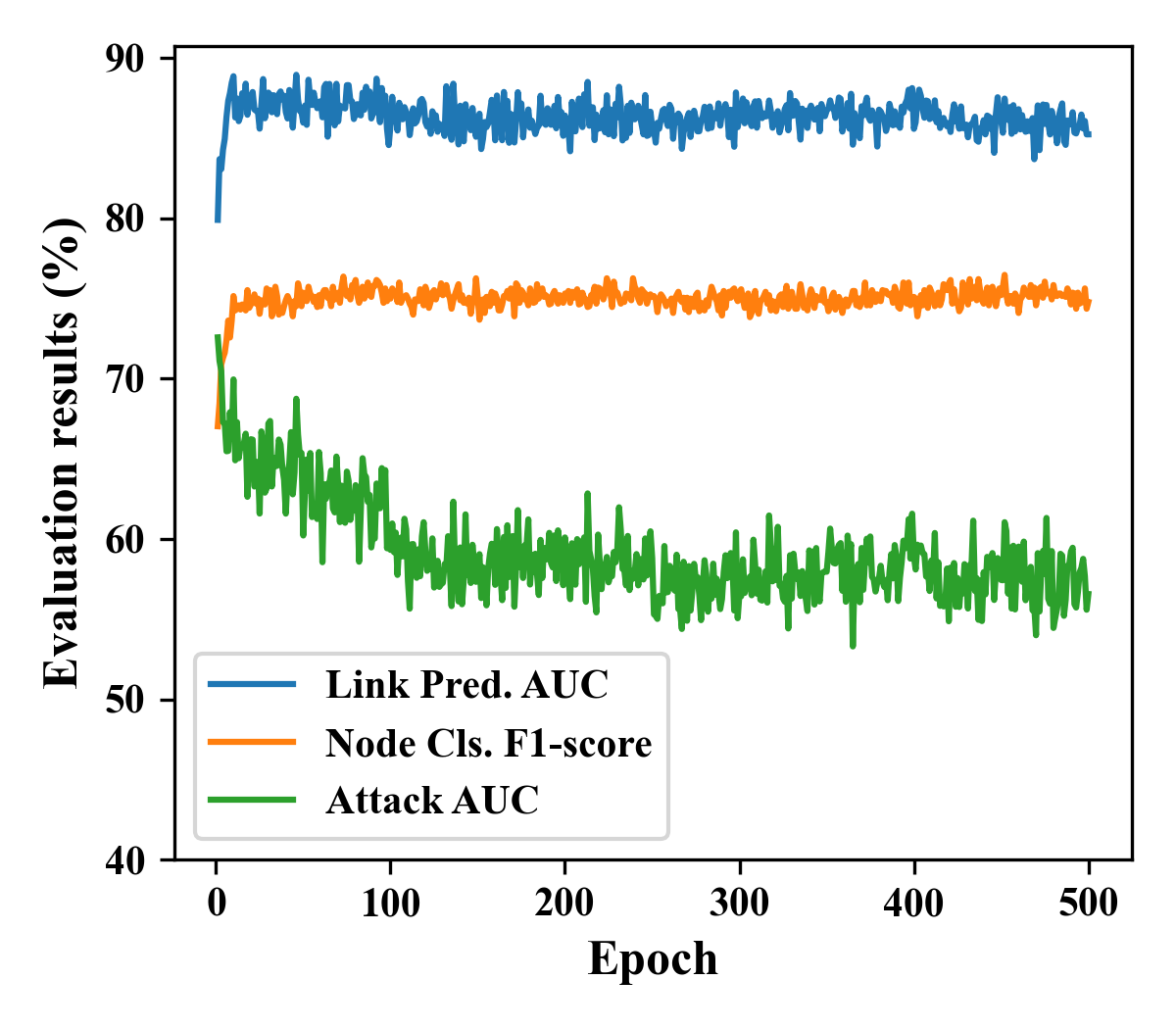}
}
\caption{Variations in loss and privacy/utility evaluation results as the number of epochs increases in the training process of graph learner $\mathcal H_\theta$ on \textit{Cora} under different $\mu$ settings.}
\label{fig:conver}
\vspace{-1em}
\end{figure}

We conduct experiments to verify the convergence of PPGSL.
Fig.~\ref{fig:conver} illustrates the changes in loss and privacy/utility evaluation results during the training of the graph learner $\mathcal H_\theta$ for different values of \(\mu\) (the update interval of the surrogate attack model).
We observe that training loss generally decreases, with minor fluctuations when $\mu\neq 1$. In particular, when $\mu=50$, the loss jumps every 50 epochs (see Fig.~\ref{fig:conver_50_loss}) due to updates of the surrogate attack model. At these points, the surrogate model changes to a stronger version, 
temporarily increasing the privacy loss.

Regarding the privacy/utility evaluation results (Fig.~\ref{fig:conver_1_eva}, \ref{fig:conver_10_eva}, and \ref{fig:conver_50_eva}), we observe that as the number of epochs increases, both privacy protection (as indicated by a general decrease in the attack AUC) and utility maintenance (reflected in the increase in the link prediction AUC and node classification F1 score) improve.
Therefore, the PPGSL demonstrates fast and stable convergence capability, ensuring the effectiveness of the privacy-preserving graph learning process. Additionally, the convergence speed is higher when $\mu$ is smaller.

\subsubsection{Comparison of Different Training Protocols}
\label{sec:comparison_sitp}

\begin{figure}[t]
\centering
\subfloat[\footnotesize{Evaluation results under SITP.}\label{fig:compare_sitp_eva}]
{
\includegraphics[width=0.45\linewidth]{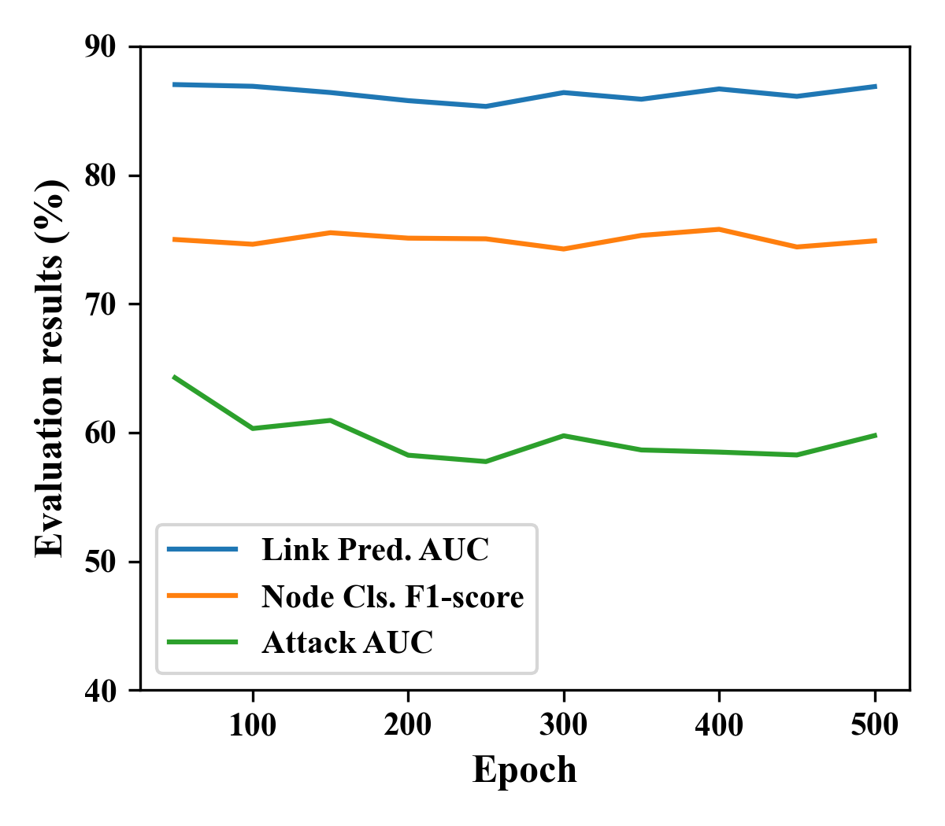}
}
\
\subfloat[\footnotesize{Evaluation results under ADV.}\label{fig:compare_gan_eva}]
{
\includegraphics[width=0.45\linewidth]{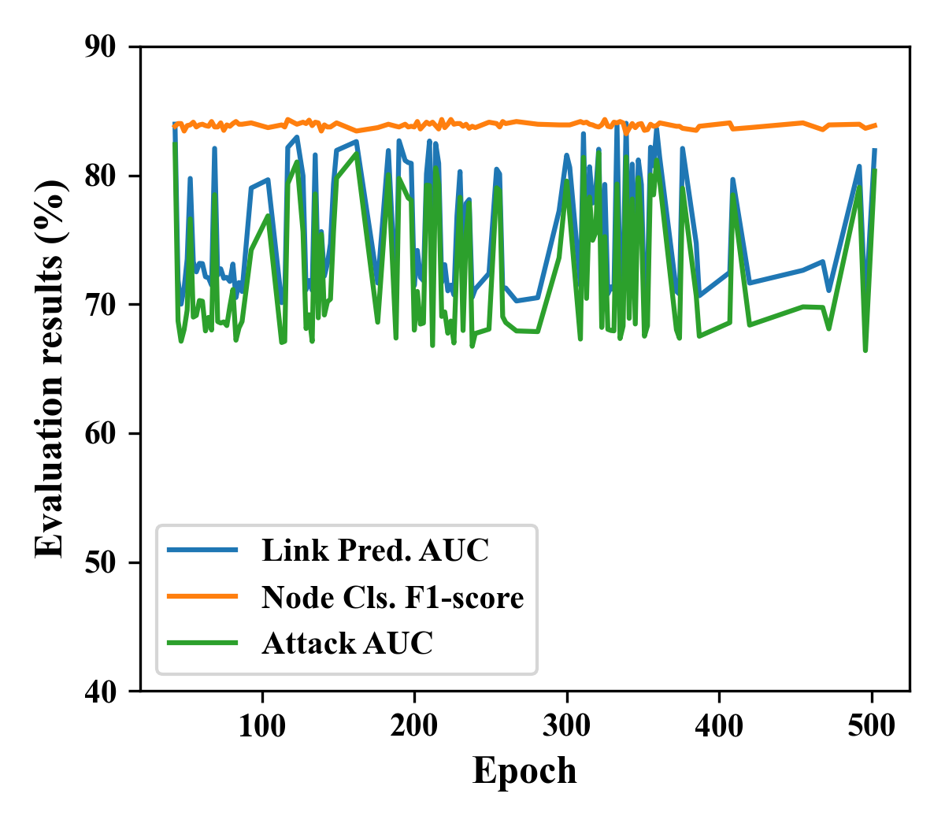}
}
\caption{Comparison of privacy/utility evaluation results across different training protocols.
SITP: our proposed secure iterative training protocol, with a surrogate attack model update interval of \(\mu=50\); ADV: typical adversarial training protocols, such as AdvReg~\cite{nasr2018machine}.}
\label{fig:compare_gan}
\vspace{-1em}
\end{figure}

We compare our proposed SITP (with an update interval of $\mu=50$) against a standard adversarial training baseline, ADV, which follows the protocol of AdvReg~\cite{nasr2018machine}. As shown in Fig.~\ref{fig:compare_gan}, SITP demonstrates a superior privacy--utility trade-off. Compared with the ADV baseline, it simultaneously achieves a higher link prediction AUC (better utility) and a lower attack AUC (stronger privacy). A detailed comparison of convergence and training efficiency is provided in Appendix~\ref{sec:comparison_sitp_appendix}.

\section{Related Works}
\subsection{Privacy-Preserving Graph Data Publishing}
We first review three typical categories of privacy-preserving graph data publishing studies aimed at protecting link privacy: graph anonymization methods, differential privacy (DP) methods, and heuristic-based methods. Graph anonymization methods~\cite{zheleva2008preserving,ying2011link,ying2008randomizing,fard2012limiting,mittalPS13,milani2015neighborhood,liu2016smartwalk} aim to anonymize the links on the original graph by rewiring them via randomization techniques while preserving certain properties of the original graph to ensure utility. However, the primary goal of anonymization is to protect links on the original graph from reidentification rather than safeguarding hidden sensitive links against inference attacks.

DP methods~\cite{YZDCCS23,nguyen2015differentially,chen2014correlated,xiao2014differentially,qin2017generating} are conventionally designed to defend against Bayesian inference attacks without arbitrary prior knowledge. In the context of graph privacy protection, DP methods usually prevent the original graph from being inferred on the basis of the released graph or embedding by introducing random noise with rigorous theoretical privacy guarantees. However, they often need to introduce excessive noise into the original data to defend against various link inference attacks (\textit{e.g.}, GNN attacks~\cite{he2021stealing,jayaraman2019evaluating}), making it challenging to achieve a good privacy--utility trade-off.

Several heuristic-based methods~\cite{han2023privacy,yu2019target} have been proposed to defend against link inference attacks. They often first assess the privacy and utility of potential perturbations in a graph; then, they manually select the perturbation that provides the best privacy--utility trade-off to find the near-optimal graph structure iteratively. However, these methods often result in significant computational costs and produce only local optimal solutions rather than global optimal ones. They barely enumerate all the potential perturbations for privacy and utility assessment. 

Compared with existing methods, the PPGSL is designed with a crafted privacy objective to measure the privacy leakage caused by sensitive link inference attacks. Additionally, the PPGSL automatically optimizes the graph structure to achieve an optimal privacy--utility trade-off via a learning method.

\subsection{Link Inference Attacks}
\label{sec:link_inference}
Link inference attacks exploit inherent patterns within the graph to reidentify or infer private structural information. Depending on the goal, method, and attacker's knowledge, these attacks can be divided into two categories: attacks in graph data publishing and attacks against open GNN APIs.

\textbf{Link inference attacks in graph data publishing}. In this context, attackers have access to the whole published graph data~\cite{yu2019target} or its node embedding~\cite{duddu2020quantifying}. Specifically, if attackers obtain the node embedding of the target graph, they can directly conduct embedding-based attacks by leveraging the similarity information between the embeddings of target node pairs~\cite{duddu2020quantifying,han2023privacy,wang2023link}. Comparatively, given the published graph data, attackers can carry out link inference attacks via Bayesian methods and structure-based approaches~\cite{zhang2018link,zhang2020towards,xian2021towards} or conduct embedding-based attacks after learning node embedding with the published graph~\cite{han2023privacy}.

\textbf{Link inference attacks against open GNN APIs}. This stream of studies typically assumes that there is an open GNN API, which is trained to complete the node classification task, and attackers have black-box access to this API~\cite{he2021stealing,wu2022linkteller,meng2023devil,zhang2023demystifying}. Recently, the LinkTeller attack method was introduced, which recovers private links on a graph through influence analysis~\cite{wu2022linkteller}. Specifically, by querying the GNN API with adversarial input node features and analyzing its influence on the target node's output, an attacker can infer whether a link exists between the input node and the target node. Moreover, some recent work has considered the disparity in individual privacy risks~\cite{zhang2023demystifying,zhang2024unraveling}.

This work focuses on protecting sensitive links from inference attacks in the context of whole graph data sharing. We also note that some existing works have proposed defense strategies regarding attacks against open GNN APIs~\cite{meng2023devil,sajadmanesh2023gap}. Since these strategies typically aim to build privacy-preserving GNN models, they are unsuitable for the graph data sharing problem we address.

\subsection{Graph Structure Learning}
Graph structure learning (GSL) seeks to simultaneously derive an optimized graph structure and corresponding graph representations~\cite{zhu2021survey}. Most current GSL studies focus on learning robust graph representations or improving the performance on specific downstream tasks through the concurrent optimization of the graph structure and the GNN encoder~\cite{chen2020iterative,luo2021learning,jin2020graph,liu2022compact,fatemi2021slaps,sun2022graph,wang2023prose}. Recently, self-supervised GSL has emerged to address scenarios where the task label is scarce or downstream tasks are unknown~\cite{liu2022towards,li2022reliable,zhao2023self}.

Although both the PPGSL and GSL procure the targeted graph via a learning-based approach, the PPGSL significantly deviates from conventional GSL methods in its training objectives and training protocol.
Specifically, while most GSL methods aim to learn a robust GNN encoder, our goal is to learn privacy-preserving graph data for publishing. The components of the PPGSL and GSL frameworks differ.
Notably, the PPGSL incorporates surrogate attack models to generate privacy protection signals---a component that is absent in standard GSL approaches. The introduction of this new component also necessitates an effective and distinct training protocol.

\section{Conclusion and Discussion}
This work makes theoretical contributions to the field of trustworthy graph systems by expanding the research scope and supplying novel privacy-preserving technologies~\cite{wu2022trustworthy}. Specifically, we formulate a novel learning-based graph data publishing problem against sensitive link inference attacks and propose a privacy-preserving graph structure learning framework, dubbed PPGSL. Two core modules are designed to parameterize the graph structure and optimize for privacy and utility objectives; a secure iterative training protocol is introduced to ensure privacy preservation and stable convergence. Theoretical analyses validate the convergence and optimality of the PPGSL. Extensive experiments demonstrate the state-of-the-art performance of the PPGSL in achieving an optimal privacy--utility trade-off.

We also discuss the limitations and future directions of this work. First, the PPGSL currently focuses only on perturbing the graph structure to protect sensitive links. Looking forward, we aim to design a unified framework that can learn both node features and topological structure for privacy-preserving data publishing.
Second, we adopt a GNN-based surrogate attack model in our method, which may not represent all adversaries with different knowledge. Devising a more powerful surrogate attack model to capture broader privacy risks is a promising future research direction.
Finally, while the PPGSL seeks to reduce the average privacy risk across the entire graph, it is crucial to consider the disparity in individual risks to ensure fairness~\cite{zhang2023demystifying,zhang2024unraveling}. Incorporating both privacy and fairness considerations can lead the PPGSL toward more comprehensive and trustworthy data sharing~\cite{zhang2024trustworthy}.







\begin{IEEEbiography}[{\includegraphics[width=1in,height=1.25in,clip,keepaspectratio]{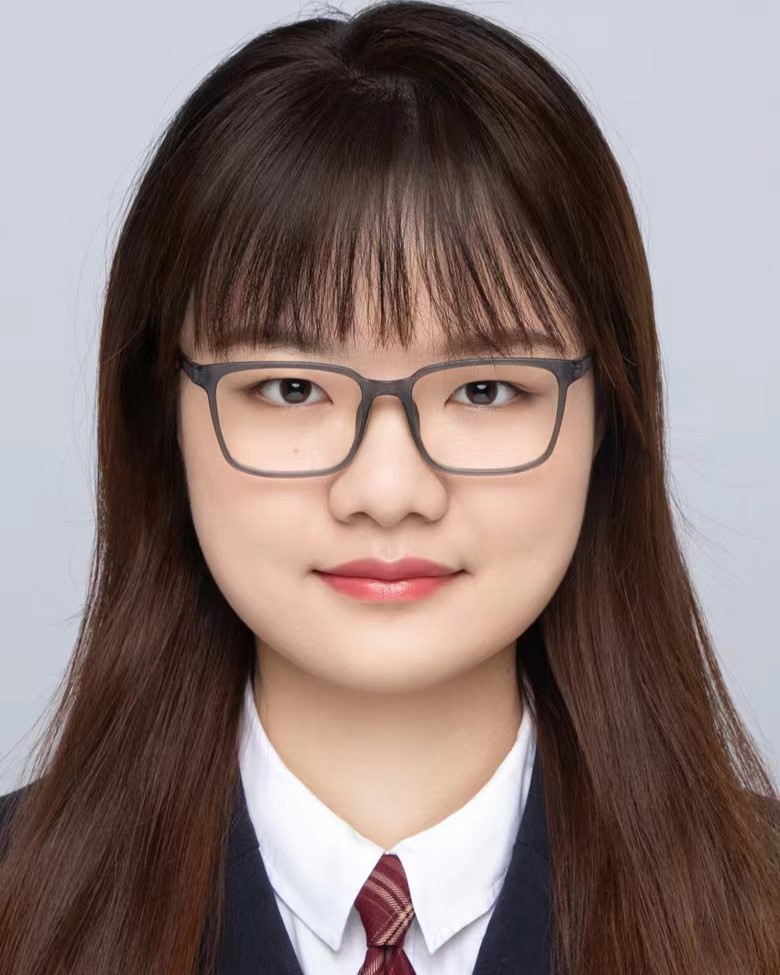}}]{Yucheng Wu}
received the bachelor's degree in data science and big data technology from Shanghai University of Finance and Economics, China, in 2023. She is currently working toward the PhD degree in computer software and theory with the School of Computer Science, Peking University, China. Her research interests include large-scale data mining and graph neural networks.
\end{IEEEbiography}

\begin{IEEEbiography}[{\includegraphics[width=1in,height=1.5in,clip,keepaspectratio]{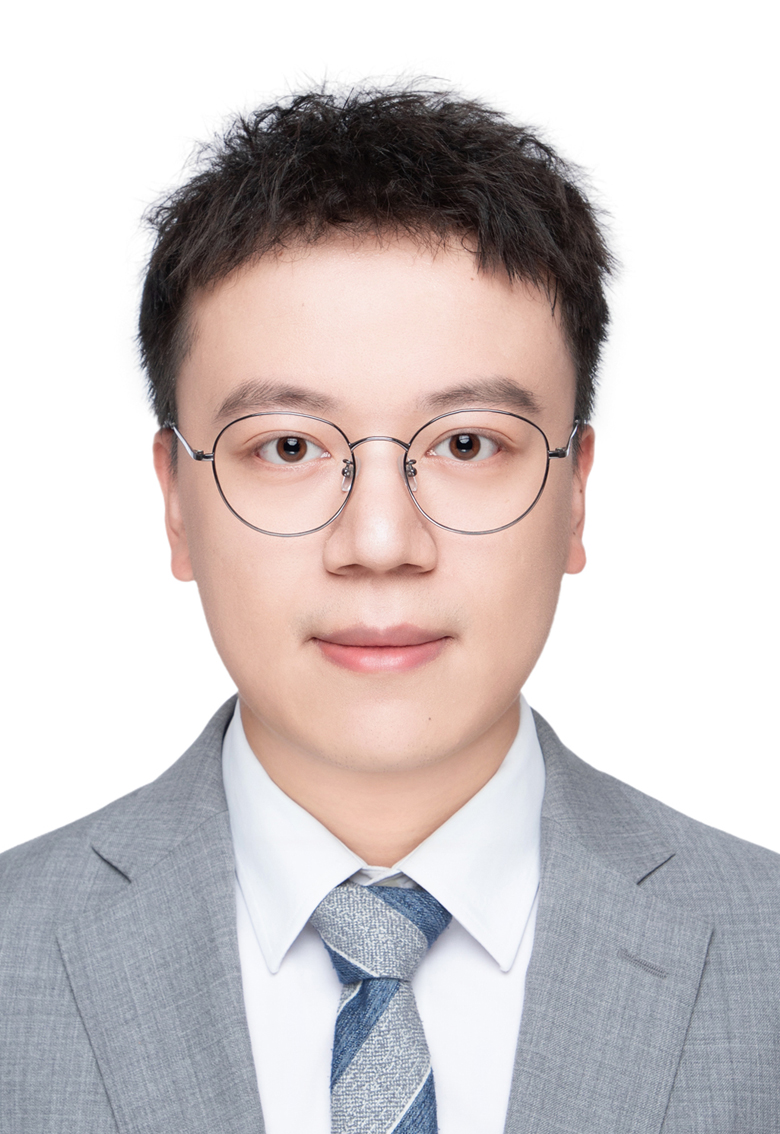}}]{Yuncong Yang}
received a BS, MS and Ph.D. degree in management science and engineering from the Shanghai University of Finance and Economics, in 2018, 2020 and 2025, respectively. He is a member of the Key Laboratory of Interdisciplinary Research of Computation and Economics (Shanghai University Finance and Economics), Ministry of Education. His work has been published in leading journals in computer science, including TDSC and TKDE. His research interests include graph neural networks and privacy protection in machine learning.
\end{IEEEbiography}

\begin{IEEEbiography}[{\includegraphics[width=1in,height=1.25in,clip,keepaspectratio]{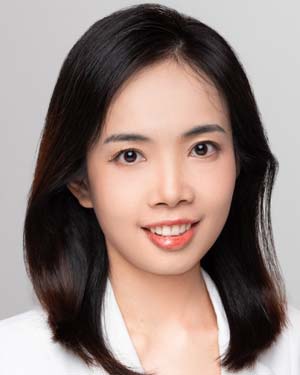}}]{Xiao Han} is a full professor at the School of Economics and Management, Beihang University. She received a Ph.D. in informatics from the Pierre and Marie Curie University and Institut Mines-TELECOM/TELECOM SudParis in 2015. Her research focuses on data-driven intelligent systems in business and societal contexts, with particular focuses on data security and privacy. Her work has been published in leading journals and top-tier conference proceedings in information systems and computer science, including MISQ, JOC, TDSC, TIFS, TKDE, TSE, WWW, AAAI, etc.
\end{IEEEbiography}

\begin{IEEEbiography}[{\includegraphics[width=1in,height=1.25in,clip,keepaspectratio]{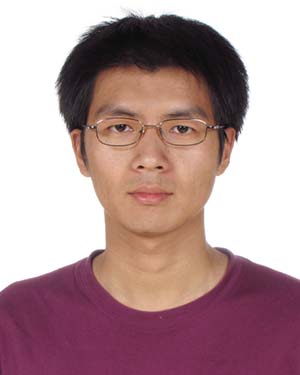}}]{Leye Wang} is a tenured associate professor at Key Lab of High Confidence Software Technologies (Peking University), Ministry of Education, China, and School of Computer Science, Peking University. His research interests include ubiquitous computing and data privacy protection. Wang received a Ph.D. in computer science from the Pierre and Marie Curie University and Institut Mines-TELECOM/TELECOM SudParis, France, in 2016. His research has appeared in journals and conference proceedings such as MISQ, JOC, TDSC, TIFS, TKDE, AI Journal, IEEE Computer, IEEE Comm. Mag., WWW, AAAI, ASE, etc.computing, mobile crowdsensing, and urban computing.
\end{IEEEbiography}

\begin{IEEEbiography}[{\includegraphics[width=1in,height=1.25in,clip,keepaspectratio]{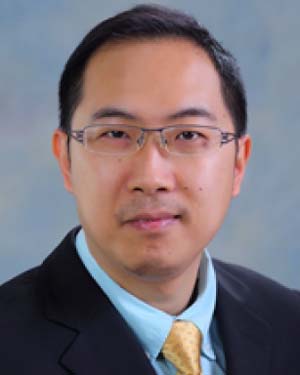}}]{Junjie Wu} is currently the full professor of the School of Economics and Management, Beihang University. He is also the director of the MIIT Key Laboratory of Data Intelligence and Management. He holds a B.E. degree from the School of Civil Engineering and a Ph.D. degree from the School of Economics and Management, Tsinghua University. His research interests include data and decision intelligence, with intense applications to business, finance, cities and industries. He has published prolifically in journals and proceedings including MISQ, ISR, JOC, TKDE, KDD, NeurIPS, AAAI, etc.
\end{IEEEbiography}



\vfill

\clearpage
\appendix

\subsection{Proof of Theoretical Analysis}
\label{appendix:proof}

\subsubsection{Convergence of the PPGSL}
\label{appendix:proof2}

\begin{lemma}
\label{lemma:stable_update}
Suppose that \(\mathcal{L}(\theta_1, \theta_2)\) is strongly convex with respect to \(\theta_1\) for any fixed \(\theta_2\), and that \(\mathcal{L}\) is continuously differentiable with respect to both \(\theta_1\) and \(\theta_2\). Then the optimal solution \(\theta_1^*\) changes slowly as \(\theta_2\) varies.
\end{lemma}

\begin{proof}
Define \(\theta_1^*(\theta_2) = \arg \min_{\theta_1} \mathcal{L}(\theta_1, \theta_2)\) as the optimal value of \(\theta_1\) given a fixed \(\theta_2\). Since \(\mathcal{L}(\theta_1, \theta_2)\) is strongly convex in \(\theta_1\) for each fixed \(\theta_2\), \(\theta_1^*(\theta_2)\) exists uniquely.

By the first-order optimality condition for \(\theta_1^*(\theta_2)\), we have:
\begin{equation}
   \nabla_{\theta_1} \mathcal{L}(\theta_1^*(\theta_2), \theta_2) = 0
\end{equation}

Since \(\mathcal{L}\) is strongly convex in \(\theta_1\), the Hessian \(\nabla_{\theta_1}^2 \mathcal{L}(\theta_1^*(\theta_2), \theta_2)\) is positive definite for each \(\theta_2\). By the implicit function theorem, \(\theta_1^*(\theta_2)\) is a continuously differentiable function of \(\theta_2\).

Taking the total derivative of \(\nabla_{\theta_1} \mathcal{L}(\theta_1^*(\theta_2), \theta_2) = 0\) with respect to \(\theta_2\), we obtain the following:
\begin{equation}
\begin{aligned}
   \frac{d}{d \theta_2} \nabla_{\theta_1} \mathcal{L}(\theta_1^*(\theta_2), \theta_2)
   &= \nabla_{\theta_1}^2 \mathcal{L}(\theta_1^*(\theta_2), \theta_2) \cdot \frac{d \theta_1^*(\theta_2)}{d \theta_2} \\
   &+ \nabla_{\theta_1} \nabla_{\theta_2} \mathcal{L}(\theta_1^*(\theta_2), \theta_2) = 0
\end{aligned}
\end{equation}

Solving for \(\frac{d \theta_1^*(\theta_2)}{d \theta_2}\), we obtain the following:
\begin{equation}
   \frac{d \theta_1^*(\theta_2)}{d \theta_2} = -\left(\nabla_{\theta_1}^2 \mathcal{L}(\theta_1^*(\theta_2), \theta_2)\right)^{-1} \nabla_{\theta_1} \nabla_{\theta_2} \mathcal{L}(\theta_1^*(\theta_2), \theta_2)
\end{equation}

Since \(\mathcal{L}\) is strongly convex with respect to \(\theta_1\), \(\nabla_{\theta_1}^2 \mathcal{L}(\theta_1^*(\theta_2), \theta_2)\) is bounded below by a positive constant \(m\), giving the following:
\begin{equation}
   \left\| \frac{d \theta_1^*(\theta_2)}{d \theta_2} \right\| \leq \frac{1}{m} \left\| \nabla_{\theta_1} \nabla_{\theta_2} \mathcal{L}(\theta_1^*(\theta_2), \theta_2) \right\|
\end{equation}

Thus, as \(\theta_2\) varies, the change in \(\theta_1^*(\theta_2)\) is bounded, implying that \(\theta_1^*(\theta_2)\) varies slowly with respect to \(\theta_2\) if \(\nabla_{\theta_1} \nabla_{\theta_2} \mathcal{L}\) is well behaved.

This bound on \(\frac{d \theta_1^*(\theta_2)}{d \theta_2}\) implies that \(\theta_1^*(\theta_2)\) changes gradually as \(\theta_2\) changes, ensuring that the updates to \(\theta_2\) do not cause large, abrupt changes in optimal \(\theta_1\). Hence, the update scheme where \(\theta_1\) is reoptimized for each update of \(\theta_2\) remains stable.
\end{proof}

This lemma provides the theoretical basis for the claim that, under strong convexity of \(\mathcal{L}\) with respect to \(\theta_1\), the optimal \(\theta_1\) will vary smoothly as \(\theta_2\) is updated, making this training approach converge more stably.

\begin{proposition}
\label{co:convergence1}
Under the training of the PPGSL with SITP, the following inequality holds from the $t$-th iteration to the $(t+1)$-th iteration:
\begin{equation}
\mathbb{E}[\mathcal{L}_{learner}(\phi^{(t+1)}, \theta^{(t+1)})] \leq\mathbb{E}[\mathcal{L}_{learner}(\phi^{(t)}, \theta^{(t)})]
\end{equation}
\end{proposition}

\begin{proof}

(1) Update Step for \(\phi\).
At each iteration \(t\), we first reinitialize \(\phi\) and optimize it with \(\theta^{(t)}\) fixed to minimize \(\mathcal{L}_{attack}(\phi, \theta^{(t)})\). Since the reinitialization is random, we consider the expected value of \(\mathcal{L}_{attack}\) after this step. Let \(\mathbb{E}[\mathcal{L}_{attack}(\phi^{(t+1)}, \theta^{(t)})]\) represent the expected minimum value of \(\mathcal{L}_{attack}\) after reinitialization and optimization over \(\phi\). By definition of \(\phi^{(t+1)}\), we have the following:
\begin{equation}
\mathbb{E}[\mathcal{L}_{attack}(\phi^{(t+1)}, \theta^{(t)})] \leq \mathbb{E}[\mathcal{L}_{attack}(\phi^{(t)}, \theta^{(t)})]
\end{equation}

As $\theta$ only updates one step in each iteration, $\mathcal{L}_{learner} (\phi,\theta)$ changes more drastically when $\theta$ changes, and the change in $\phi$ is relatively small when $\theta$ changes according to Lemma~\ref{lemma:stable_update}. Thus, the change in $\phi$ results in little change in $\mathcal{L}_{learner} (\phi,\theta)$, and we see it as not changing:
\begin{equation}
\label{eq:learner1}
\mathbb{E}[\mathcal{L}_{learner}(\phi^{(t+1)}, \theta^{(t)})] \;\dot{=}\; \mathbb{E}[\mathcal{L}_{learner}(\phi^{(t)}, \theta^{(t)})]
\end{equation}

(2) Update Step for \(\theta\). Suppose that \(\mathcal{L}_{learner}\) is convex in \(\theta\) with a Lipschitz continuous gradient.
We fix \(\phi^{(t+1)}\) and update \(\theta\) via a gradient descent step to minimize \(\mathcal{L}_{learner}(\phi^{(t+1)}, \theta)\). Assuming that \(\mathcal{L}_{learner}\) is convex in \(\theta\) and has a Lipschitz continuous gradient, we can ensure that for a sufficiently small learning rate \(\eta\), the update satisfies the following:
\begin{equation}
\begin{aligned}
\mathcal{L}_{learner}(\phi^{(t+1)}, \theta^{(t+1)}) &\leq \mathcal{L}_{learner}(\phi^{(t+1)}, \theta^{(t)}) \\
&- \frac{\eta}{2} \|\nabla_{\theta} \mathcal{L}_{learner}(\phi^{(t+1)}, \theta^{(t)})\|^2
\end{aligned}
\end{equation}

To derive this, we apply the Lipschitz continuity property to expand \(\mathcal{L}_{learner}(\phi^{(t+1)}, \theta^{(t+1)})\) around \(\theta^{(t)}\) as follows:
\begin{equation}
\begin{aligned}
\mathcal{L}_{learner}(\phi^{(t+1)}, \theta^{(t+1)}) &\leq \mathcal{L}_{learner}(\phi^{(t+1)}, \theta^{(t)}) \\
&+ \nabla_{\theta} \mathcal{L}_{learner}(\phi^{(t+1)}, \theta^{(t)})^\top (\theta^{(t+1)} - \theta^{(t)}) \\
&+ \frac{L}{2} \|\theta^{(t+1)} - \theta^{(t)}\|^2
\end{aligned}
\end{equation}
where \( L \) represents the Lipschitz constant of the gradient of the function \( \mathcal{L}_{learner}(\phi, \theta) \) with respect to \(\theta\). This means that the gradient of \(\mathcal{L}_{learner}\) with respect to \(\theta\) does not change too quickly, which we can mathematically state as
$\|\nabla_{\theta} \mathcal{L}_{learner}(\phi, \theta') - \nabla_{\theta} \mathcal{L}_{learner}(\phi, \theta)\| \leq L \|\theta' - \theta\|$
for any \(\theta\) and \(\theta'\).

Substitute the gradient descent update \(\theta^{(t+1)} = \theta^{(t)} - \eta \nabla_{\theta} \mathcal{L}_{learner}(\phi^{(t+1)}, \theta^{(t)})\) into the following inequality:
\begin{equation}
\begin{aligned}
\mathcal{L}_{learner}(\phi^{(t+1)}, \theta^{(t+1)}) &\leq \mathcal{L}_{learner}(\phi^{(t+1)}, \theta^{(t)}) \\
&- \eta \|\nabla_{\theta} \mathcal{L}_{learner}(\phi^{(t+1)}, \theta^{(t)})\|^2 \\
&+ \frac{L \eta^2}{2} \|\nabla_{\theta} \mathcal{L}_{learner}(\phi^{(t+1)}, \theta^{(t)})\|^2
\end{aligned}
\end{equation}

The terms can be rearranged to obtain the following:
\begin{equation}
\begin{aligned}
\mathcal{L}_{learner}(\phi^{(t+1)}, \theta^{(t+1)}) &\leq \mathcal{L}_{learner}(\phi^{(t+1)}, \theta^{(t)}) \\
&- \left(\eta - \frac{L \eta^2}{2}\right) \|\nabla_{\theta} \mathcal{L}_{learner}(\phi^{(t+1)}, \theta^{(t)})\|^2
\end{aligned}
\end{equation}

Choosing \(\eta \leq \frac{1}{L}\) ensures that \(\eta - \frac{L \eta^2}{2} \geq \frac{\eta}{2}\), yielding
\begin{equation}
\begin{aligned}
\mathcal{L}_{learner}(\phi^{(t+1)}, \theta^{(t+1)}) &\leq \mathcal{L}_{learner}(\phi^{(t+1)}, \theta^{(t)})\\
&- \frac{\eta}{2} \|\nabla_{\theta} \mathcal{L}_{learner}(\phi^{(t+1)}, \theta^{(t)})\|^2
\end{aligned}
\end{equation}

Thus, as expected, updating \(\theta\) decreases \(\mathcal{L}_{learner}\) as follows:
\begin{equation}
\begin{aligned}
\label{eq:learner2}
\mathbb{E}[\mathcal{L}_{learner}(\phi^{(t+1)}, \theta^{(t+1)})]
\leq\mathbb{E}[\mathcal{L}_{learner}(\phi^{(t+1)}, \theta^{(t)})]
\end{aligned}
\end{equation}

(3) Combining Updates.
Combining Eq.~\ref{eq:learner1} and Eq.~\ref{eq:learner2}, the following inequality holds from the $t$-th iteration to the $(t+1)$-th iteration:
\begin{equation}
\mathbb{E}[\mathcal{L}_{learner}(\phi^{(t+1)}, \theta^{(t+1)})] \leq\mathbb{E}[\mathcal{L}_{learner}(\phi^{(t)}, \theta^{(t)})]
\end{equation}

\end{proof}

\subsubsection{Optimal Privacy--utility Trade-off of the PPGSL}

\begin{proposition}
\label{theorem:lagrange_dual}
Consider the following two optimization problems:
\begin{itemize}[leftmargin=5.5mm]
    \item Constrained Optimization Problem:
    \begin{equation}
        \min_\theta \, \mathcal L_{priv}(\theta) \quad \text{s.t.} \quad  \mathcal L_{util}(\theta) \leq \epsilon
    \end{equation}
    \item Unconstrained Regularization Problem:
    \begin{equation}
         \min_\theta \, \mathcal L_{learner}(\theta) = \mathcal L_{priv}(\theta) + \alpha \, \mathcal L_{util}(\theta)
    \end{equation}
\end{itemize}

Assume that both $\mathcal L_{priv}(\theta)$ and $\mathcal L_{util}(\theta)$ are convex functions of $\theta$ and that the feasible region satisfies Slater's condition. Then, there must exist an \(\alpha^* \geq 0\) such that if \(\alpha = \alpha^*\), the optimal solutions of both the constrained and unconstrained problems are equivalent.

\end{proposition}

\begin{proof}

We begin by constructing the Lagrangian for the constrained problem with the Lagrange multiplier $\lambda \geq 0$:
\begin{equation}
    \mathcal{L}_{lagr}(\theta, \lambda) = \mathcal{L}_{priv}(\theta) + \lambda \left( \mathcal{L}_{util}(\theta) - \epsilon \right)
\end{equation}
Our goal is to solve the following:
\begin{equation}
   \max_{\lambda \geq 0} \min_\theta \, \mathcal{L}_{lagr}(\theta, \lambda)
\end{equation}

Karush--Kuhn--Tucker (KKT) Conditions: Given that \(\mathcal{L}_{priv}(\theta)\) and \(\mathcal{L}_{util}(\theta)\) are convex functions and that the feasible region satisfies Slater's condition, the KKT conditions are both necessary and sufficient for optimality. Let \(\theta^*\) be the solution to the constrained problem. The KKT conditions for \((\theta^*, \lambda^*)\) are as follows:
\begin{itemize}[leftmargin=5.5mm]
\item Primal feasibility: \(\mathcal{L}_{util}(\theta^*) \leq \epsilon\).
\item Dual feasibility: \(\lambda^* \geq 0\).
\item Stationarity: \(\nabla_\theta \mathcal{L}_{priv}(\theta^*) + \lambda^* \nabla_\theta \mathcal{L}_{util}(\theta^*) = 0\).
\item Complementary slackness: \(\lambda^* \left( \mathcal{L}_{util}(\theta^*) - \epsilon \right) = 0\).
\end{itemize}

By the complementary slackness condition, we consider two possible cases:
\begin{itemize}[leftmargin=5.5mm]
\item Case 1: \(\mathcal{L}_{util}(\theta^*) = \epsilon\). Here, the constraint is active, meaning that \(\mathcal{L}_{util}(\theta^*)\) reaches the upper limit \(\epsilon\). The KKT conditions guarantee the existence of \(\lambda^* > 0\). On the basis of the stationarity condition, $\lambda^* = -\frac{\nabla_\theta \mathcal{L}_{priv}(\theta^*)}{\nabla_\theta \mathcal{L}_{util}(\theta^*)}$.
The solution \(\theta^*\) for the constrained problem also minimizes the objective $\mathcal{L}_{learner}(\theta) = \mathcal{L}_{priv}(\theta) + \alpha \mathcal{L}_{util}(\theta)$ when \(\alpha^* = -\frac{\nabla_\theta \mathcal{L}_{priv}(\theta^*)}{\nabla_\theta \mathcal{L}_{util}(\theta^*)}\), as $\nabla_\theta \mathcal{L}_{priv}(\theta^*) + \alpha^* \nabla_\theta \mathcal{L}_{util}(\theta^*) = 0$ for the extreme point \(\theta^*\).

\item Case 2: \(\lambda^* = 0\). Here, the utility constraint is not binding, \textit{i.e.}, \(\mathcal{L}_{util}(\theta^*) < \epsilon\). Since \(\lambda^* = 0\), minimizing the unconstrained problem is equivalent to minimizing \(\mathcal{L}_{priv}(\theta)\) alone. This solution also satisfies the original constraint \(\mathcal{L}_{util}(\theta^*) \leq \epsilon\) without requiring any additional penalty. Furthermore, for optimizing \(\mathcal{L}_{learner}(\theta)\), minimizing \(\mathcal{L}_{priv}(\theta)\) alone is equivalent to having \(\alpha^* = 0\).

\end{itemize}

Therefore, under these conditions, there exists an \(\alpha = \alpha^*\) such that the optimal solution \(\theta^*\) of the constrained problem also minimizes the unconstrained regularization problem. The two problems are equivalent when:
\begin{equation}
   \alpha = \alpha^*, \quad \text{where } \alpha^* \text{ satisfies the KKT conditions.}
\end{equation}

\end{proof}

\subsubsection{Generalized Privacy Protection Performance of the PPGSL}
The GNN encoder $f_\phi$, parameterized by $\phi$, maps $x_u$ in the observational space to an embedding vector $z_u$ in a latent space, \textit{i.e.}, $z_i=f_\phi (x_i)$, and $Z = [z_1; z_2;\ldots; z_N] \in\mathbb R^{N\times d}$ is the node embedding matrix. $Z'$ is the node embedding matrix of $G'$, where $Z'=f_{\phi}(G')$. Here, $f_\phi$ is the GNN encoder of the surrogate attack model. Taking node embedding as a bridge, we can reformulate the privacy goal $\min_{\theta} I(G'; E_s)$ as follows\footnote{Note that we slightly misuse the notations $G$ and $E_s$ for the sake of simplicity. They are initially employed to represent the original graph and the sensitive links, and we also use them to denote random variables here.}:
\begin{subequations}
\label{eq:mutual_info2}
\begin{align}
\label{eq:mutual_info2_i}
 \text{Embedding Goal: }  \max_{\phi} I(G';Z')\\
\label{eq:mutual_info2_p}
  \text{New Privacy Goal: } \min_\theta I(Z';E_s)
\end{align}
\end{subequations}

Since the random variables in Eq.~\ref{eq:mutual_info2} are possibly high-dimensional and their posterior distributions are unknown, the mutual information terms are difficult to calculate. Motivated by existing mutual information neural estimation methods~\cite{alemi2017variational,belghazi2018mutual,cheng2020club,hjelm2018learning,poole2019variational,wang2021privacy}, we solve this challenge by translating the intractable mutual information terms into tractable terms by designing variational bounds.

\begin{lemma}
\label{le:intermediate_goal}
Training neural networks $f_{\phi}$ with the objective function $\mathcal{L}_{attack}$ of the PPGSL (Eq.~\ref{eq:gnn_encoder}) is equivalent to achieving the embedding goal (Eq.~\ref{eq:mutual_info2_i}).
\end{lemma}

\begin{proof}
For the mutual information term in Eq.~\ref{eq:mutual_info2_i}, we derive the following variational lower bound:
\begin{equation}
\begin{aligned}
    &I(G';Z')\\
    =& I(w'_{ij};z'_i,z'_j) \\
    =& H(w'_{ij})-H(w'_{ij}|z'_i,z'_j) \\
    =& H(w'_{ij})+\mathbb E_{p(z'_i,z'_j,w'_{ij})}\log p(w'_{ij}|z'_i,z'_j) \\
    =& H(w'_{ij})+\mathbb E_{p(z'_i,z'_j,w'_{ij})}\,\textit{KL}\left (p(\cdot|z'_i,z'_j)||q_\varphi(\cdot|z'_i,z'_j) \right ) \\
    &\quad\quad\quad\quad + \mathbb E_{p(z'_i,z'_j,w'_{ij})}\log q_\varphi (w'_{ij}|z'_i,z'_j) \\
    \ge& H(w'_{ij})+ \mathbb E_{p(z'_i,z'_j,w'_{ij})}\log q_\varphi (w'_{ij}|z'_i,z'_j) \\
    :=&I_{\textit{vLB}}(w'_{ij};z'_i,z'_j)
\end{aligned}
\label{eq:mutual_info15}
\end{equation}
where $\textit{KL}\left (p(\cdot)||q(\cdot) \right )$ is the Kullback--Leibler divergence between two distributions $p(\cdot)$ and $q(\cdot)$ and is nonnegative. $q_\varphi$ is an (arbitrary) auxiliary posterior distribution. $I_{\textit{vLB}}(w'_{ij};z'_i,z'_j)$ is the variational lower bound of the mutual information term, and $H(w'_{ij})$ is a constant. Note that the lower bound is tight when the auxiliary distribution $q_\varphi$ becomes the true posterior distribution $p$.

Our goal is to maximize the variational lower bound by estimating the auxiliary posterior distribution $q_\varphi$ via a parameterized neural network.
We parameterize $q_\varphi$ via a link predictor $\delta_\varphi$ defined on the node representations.
Specifically, we have
\begin{equation}
\begin{aligned}
    &\max_{\phi}\max_\varphi I_{\textit{vLB}}(w'_{ij};z'_i,z'_j) \\
    \Leftrightarrow &\max_{\phi}\max_\varphi \mathbb E_{p(z'_i,z'_j,w'_{ij})}\log q_\varphi (w'_{ij}|z'_i,z'_j) \\
    = &\max_{\phi}\max_\varphi \mathbb E_{p(z'_i,z'_j,w'_{ij})}\log q_\varphi (w'_{ij}|f_{\phi}(x'_i),f_{\phi}(x'_j)) \\
    \approx &\max_{\phi}\max_\varphi \sum_{\langle i,j\rangle\in E_{p}'\cup E_{n}'} -\textit{CE}\, (\delta_\varphi (f_{\phi}(x'_i),f_{\phi}(x'_j)),w'_{ij}) \\
    \Leftrightarrow &\min_{\phi}\min_\varphi \sum_{\langle i,j\rangle\in E_{p}'\cup E_{n}'} \textit{CE}\, (\delta_\varphi (f_{\phi}(x'_i),f_{\phi}(x'_j)),w'_{ij})
\end{aligned}
\label{eq:mutual_info16}
\end{equation}
where $\textit{CE}\,(\cdot, \cdot)$ denotes the cross-entropy function. In the learned graph $G'$, $E_{p}'$ is a set of sampled existing edges, and $E_{n}'$ is a set of sampled nonexistent edges. The final objective function in Eq.~\ref{eq:mutual_info16} is just $\mathcal{L}_{attack}$ in Eq.~\ref{eq:gnn_encoder}.
Therefore, by taking $\mathcal{L}_{attack}$ as an objective function and training the parameterized neural networks, we can achieve the embedding goal (Eq.~\ref{eq:mutual_info2_i}).
\end{proof}

\begin{lemma}
\label{le:privacy_goal}
Training neural networks $\mathcal H_{\theta}$ with the objective function $\mathcal{L}_{priv}$ of the PPGSL (Eq.~\ref{eq:embedding_dis}) is equivalent to achieving the new privacy goal (Eq.~\ref{eq:mutual_info2_p}).
\end{lemma}

\begin{proof}
We define $w^s_{ij} \in \{0,1\}$ as the adjacency matrix of sensitive links, where $w^s_{ij} = 1$ if $\langle v_i,v_j\rangle\in E_s$ and $w^s_{ij} = 0$ otherwise.
The new privacy goal (Eq.~\ref{eq:mutual_info2_p}) can be specified as $\min_\theta I(w_{ij}^{s};z'_i,z'_j)$, and we derive the vCLUB inspired by \cite{cheng2020club} as follows:
\begin{equation}
\begin{aligned}
    &I(E_s;Z')\\
    =& I(w_{ij}^{s};z'_i,z'_j) \\
    \le &I_{\textit{vCLUB}}(w_{ij}^{s};z'_i,z'_j) \\
    = &\mathbb E_{p(z'_i,z'_j,w_{ij}^{s})}\log q_\varphi (w_{ij}^{s}|z'_i,z'_j) \\
    &\quad - \mathbb E_{p(z'_i,z'_j)p(w_{ij}^{s})}\log q_\varphi (w_{ij}^{s}|z'_i,z'_j)\\
\end{aligned}
\end{equation}
where $q_\varphi (w_{ij}^{s}|z'_i,z'_j)$ is an auxiliary distribution of $p(w_{ij}^{s}|z'_i,z'_j)$ that needs to satisfy the following condition:
\begin{equation}
\begin{aligned}
    &\textit{KL} \left(p(z'_i,z'_j,w_{ij}^{s}) || q_\varphi(z'_i,z'_j,w_{ij}^{s}) \right ) \le \\
    &\quad\quad\quad\quad \textit{KL} \left(p(z'_i,z'_j)p(w_{ij}^{s}) || q_\varphi(z'_i,z'_j,w_{ij}^{s}) \right )
\end{aligned}
\end{equation}
That is, $I_{\textit{vCLUB}}$ is a mutual information upper bound if the variational joint distribution $q_\varphi(z'_i,z'_j,w_{ij}^{s})$ is closer to the joint distribution $p(z'_i,z'_j,w_{ij}^{s})$ than to $p(z'_i,z'_j)p(w_{ij}^{s})$.

To achieve the above inequation, we need to minimize the KL-divergence as follows:
\begin{equation}
\begin{aligned}
    &\min_\varphi \textit{KL} \left(p(z'_i,z'_j,w_{ij}^{s}) || q_\varphi(z'_i,z'_j,w_{ij}^{s}) \right ) \\
    =&\min_\varphi \textit{KL} \left(p(w_{ij}^{s}|z'_i,z'_j) || q_\varphi(w_{ij}^{s}|z'_i,z'_j) \right ) \\
    =&\min_\varphi \mathbb E_{p(z'_i,z'_j,w_{ij}^{s})}\log p (w_{ij}^{s}|z'_i,z'_j) \\
    &\quad\quad\quad\quad - \mathbb E_{p(z'_i,z'_j,w_{ij}^{s})}\log q_\varphi (w_{ij}^{s}|z'_i,z'_j)\\
    \Leftrightarrow & \max_\varphi \mathbb E_{p(z'_i,z'_j,w_{ij}^{s})}\log q_\varphi (w_{ij}^{s}|z'_i,z'_j)
\end{aligned}
\label{eq:mutual_info27}
\end{equation}

Finally, our target to achieve Eq.~\ref{eq:mutual_info2_p} becomes the following adversarial training objective:
\begin{equation}
\begin{aligned}
    &\min_\theta\min_\varphi I_{\textit{vCLUB}}(w_{ij}^{s};z'_i,z'_j)\\
    \Leftrightarrow &\min_\theta\max_\varphi \mathbb E_{p(z'_i,z'_j,w_{ij}^{s})}\log q_\varphi (w_{ij}^{s}|z'_i,z'_j) \\
    \approx &  \min_{\theta} \max_{\varphi} \sum_{\langle i,j\rangle\in E_{s}} -\textit{CE}\, (\delta_\varphi (z'_i,z'_j),w^{s}_{ij})\\
    \approx &  \min_{\theta} \max_{\varphi} \sum_{\langle i,j\rangle\in E_{s}} -\textit{CE}\, (\delta_\varphi (z'_i,z'_j),1)
\end{aligned}
\label{eq:mutual_info28}
\end{equation}
where $\textit{CE}\,(\cdot, \cdot)$ denotes the cross-entropy function. Eq.~\ref{eq:mutual_info28} draws away the sensitive node pair embeddings, which is equivalent to our training objective $\mathcal L_{priv}$.
Therefore, by taking $\mathcal L_{priv}$ as an objective function and training the parameterized neural networks, we can achieve the new privacy goal Eq.~\ref{eq:mutual_info2_p}.
\end{proof}

\begin{lemma}
\label{the:mutual_info1}
In every update step of \(\mathcal H_\theta\), if \(f_\phi\) is sufficiently retrained to convergence, then optimizing $\theta$ to minimize \(I(Z';E_s)\) results in a decrease in \(I(G';E_s)\).
\end{lemma}

\begin{proof}
In each update step of the graph learner, we fully retrain \(f_\phi\) until convergence. According to Lemma~\ref{le:intermediate_goal}, training \(f_\phi\) with the objective function \(\mathcal{L}_{attack}\) will achieve the intermediate goal (Eq.~\ref{eq:mutual_info2_p}), \textit{i.e.}, \(I(G';Z')\) reaches its maximum value. Owing to the expressive limitations of GNN~\cite{xu2018how}, we assume that the maximum value of \(I(G';Z')\) remains unchanged.
We prove that if \(I(G';Z')\) does not change and \(I(Z';E_s)\) decreases, then \(I(G';E_s)\) will also decrease.

Let \(H(\cdot)\) denote the information entropy of a random variable. Given that the variability space of the random variables \(G'\), \(E_s\), and \(Z'\) is fixed, we assume that \(H(G')\), \(H(E_s)\), and \(H(Z')\) remain constant.
From the relationship of mutual information, we have the following:
\begin{equation}
I(Z';E_s) = I(Z';E_s|G') + I(Z';E_s;G')
\end{equation}

Assuming that \(I(Z';E_s;G')\) either increases or remains unchanged, we can deduce that \(I(Z';E_s|G')\) must decrease. Since \(I(Z';E_s|G')\) is bounded by \(0 \leq I(Z';E_s|G') \leq H(Z'|G')\), it cannot decrease indefinitely. After sufficient iterations of optimization, \(I(Z';E_s|G')\) approaches zero and cannot decrease further. At this point, the assumption that \(I(Z';E_s|G')\) decreases becomes invalid, leading to the conclusion that \(I(Z';E_s;G')\) must decrease.

During the overall optimization process, since \(I(G';Z')\) remains unchanged and \(I(Z';E_s)\) decreases, we conclude that \(I(Z';E_s;G')\) decreases. Consequently, we find that $I(E_s;G') = I(G';E_s|Z') + I(Z';E_s;G')$ decreases, thus confirming that the derivation holds.
\end{proof}

\begin{proposition}
\label{co:fano}
In PPGSL training process, the mutual information $I(G';E_s)$ decreases.
\end{proposition}

\begin{proof}
From Lemmas~\ref{le:intermediate_goal} and~\ref{le:privacy_goal}, training the objective functions in the PPGSL increases $I(G';Z')$ and decreases $I(Z';E_s)$.
According to Lemma~\ref{the:mutual_info1}, increasing \( I(G'; Z') \) and decreasing \( I(Z'; E_s) \) together lead to a decrease in \(I(G'; E_s) \) in the PPGSL training process.
\end{proof}

\subsection{Guidance on Trade-off Parameter Selection}
\label{sec:alpha_Selection}
Empirically, the trade-off parameter $\alpha$ is typically chosen within the range of [0.001, 0.01]. The appropriate value of $\alpha$ tends to be smaller for larger graphs (\textit{i.e.,} graphs with more nodes) to achieve a similar level of privacy protection. This is because in larger graphs, the adjacency matrix has more entries. Consequently, for the same proportional perturbation, the Frobenius norm of the difference in adjacency matrices, and thus the utility loss $\mathcal{L}_{util}$, will be larger. A smaller $\alpha$ is therefore needed to balance $\mathcal{L}_{util}$ with the privacy loss term $\mathcal{L}_{priv}$. Moreover, to adjust the balance between utility and privacy, if a higher level of utility is desired, $\alpha$ can be appropriately increased; if a stronger privacy protection effect is needed, $\alpha$ can be appropriately decreased.

\subsection{Comparison of Different Training Protocols}
\label{sec:comparison_sitp_appendix}

\begin{figure}[t]
\centering

\subfloat[\footnotesize{Graph learner training loss under SITP.}\label{fig:compare_sitp_learner}]
{
\includegraphics[width=0.45\linewidth]{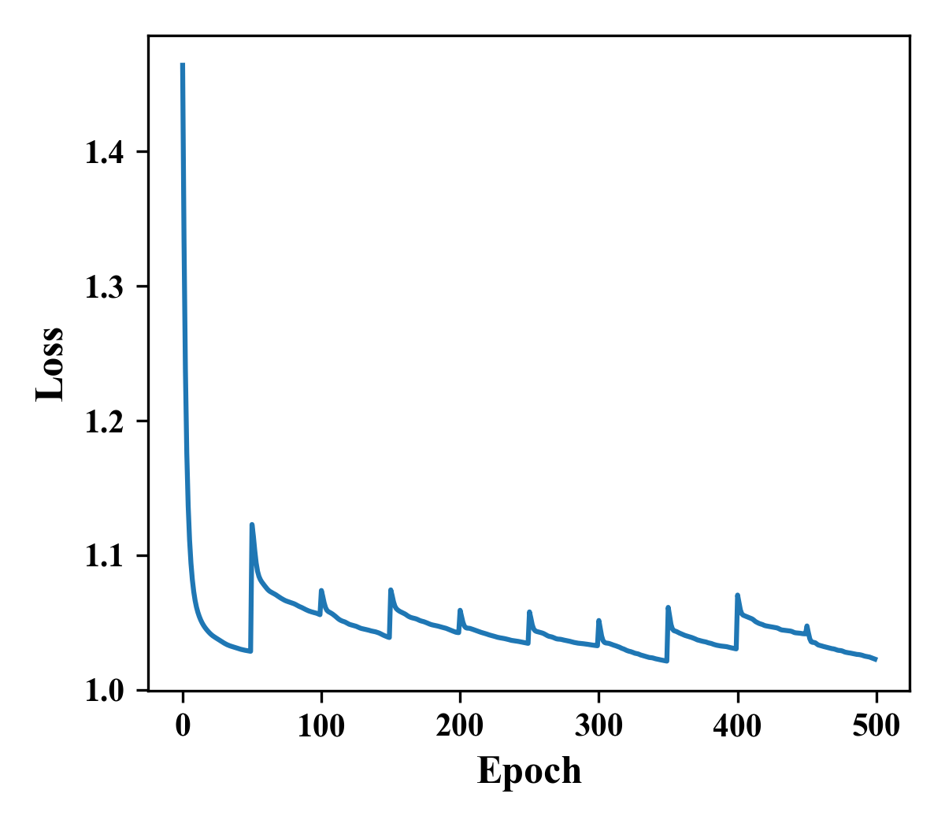}
}
\
\subfloat[\footnotesize{Graph learner training loss under ADV.}\label{fig:compare_gan_learner}]
{
\includegraphics[width=0.45\linewidth]{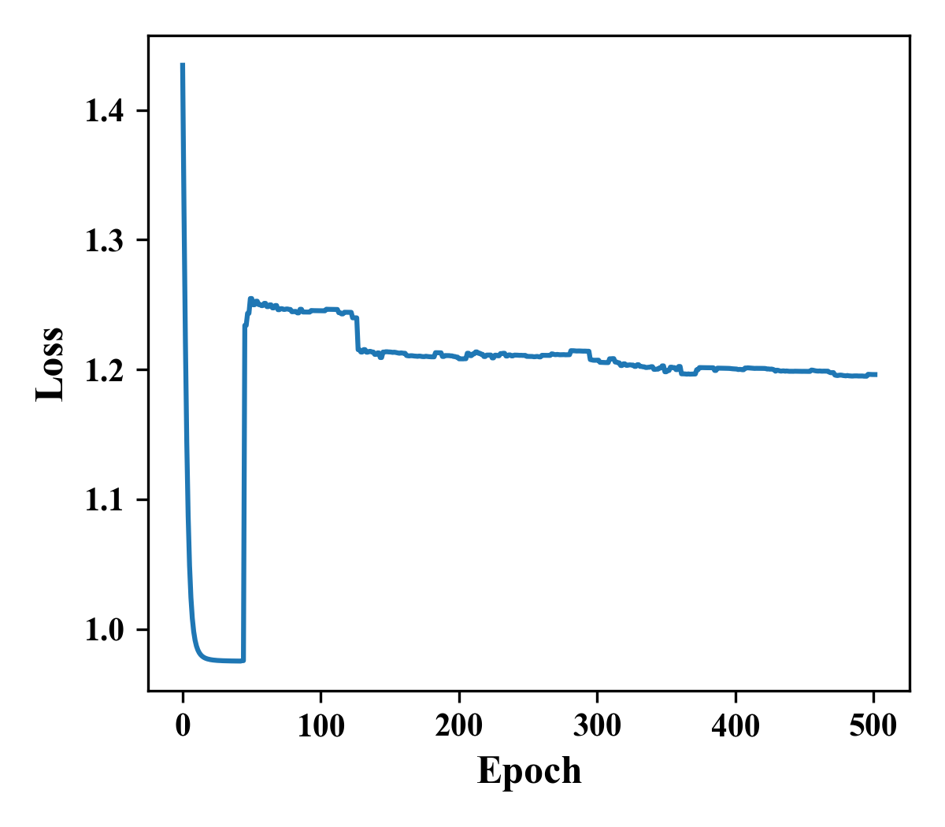}
}
\vspace{-0.8em}

\subfloat[\footnotesize{Surrogate attack model training loss under SITP.}\label{fig:compare_sitp_attack}]
{
\includegraphics[width=0.45\linewidth]{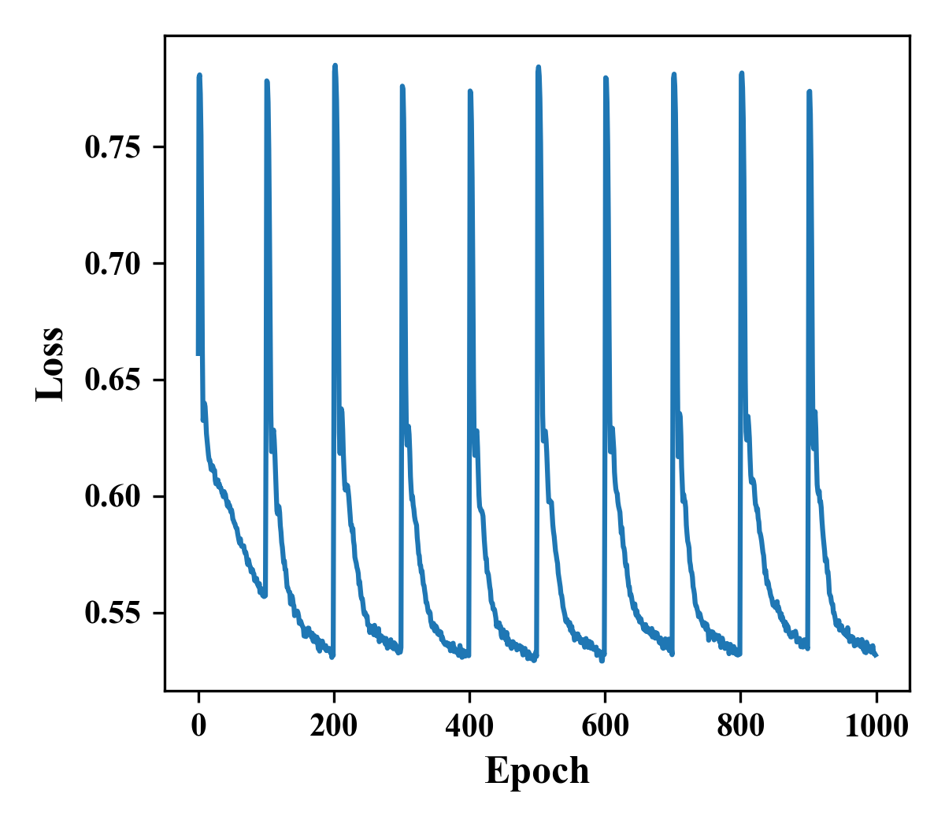}
}
\
\subfloat[\footnotesize{Surrogate attack model training loss under ADV.}\label{fig:compare_gan_attack}]
{
\includegraphics[width=0.45\linewidth]{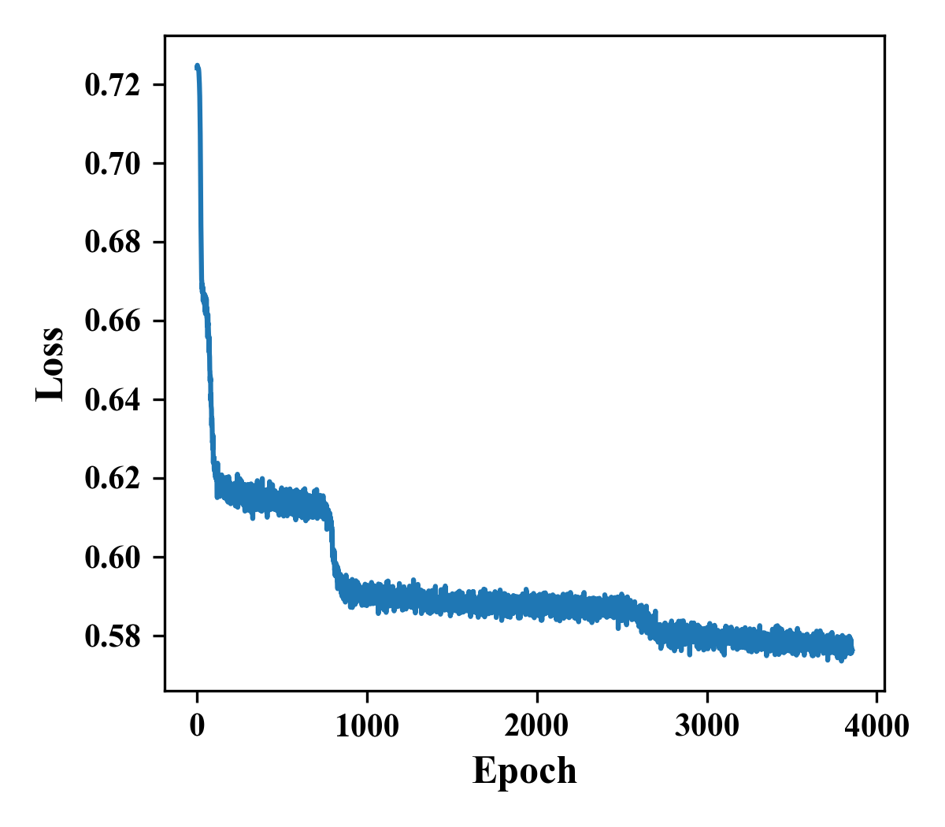}
}
\caption{Comparison of training robustness across different training protocols.
SITP: our proposed secure iterative training protocol, with a surrogate attack model update interval of \(\mu=50\); ADV: typical adversarial training protocols, such as AdvReg~\cite{nasr2018machine}.}
\label{fig:compare_gan_appendix}
\vspace{-1em}
\end{figure}

We compare the training robustness between our proposed SITP with an update interval $\mu=50$ and typical adversarial training protocols similar to AdvReg~\cite{nasr2018machine} (denoted ADV).
From Fig.~\ref{fig:compare_gan_appendix}, we can conclude that SITP outperforms ADV in terms of convergence and training efficiency.
\begin{itemize}[leftmargin=5.5mm]
    \item Convergence: Under SITP, the graph learner's training loss generally decreases over time, with minor jumps at multiples of 50 epochs due to updates in the surrogate attack model. In contrast, under ADV, the graph learner's training loss does not consistently decrease, which may converge to a local rather than a global optimum (compare Fig.~\ref{fig:compare_sitp_learner} and~\ref{fig:compare_gan_learner}).
    \item Training efficiency: With SITP, when the graph learner updates for 500 epochs, the surrogate attack model only requires 1,000 total updates (due to its faster convergence, where retraining does not significantly increase the computation time). In contrast, ADV needs approximately 4,000 updates of the surrogate attack model for the same number of graph learner updates, resulting in higher training overhead (compare Fig.~\ref{fig:compare_sitp_attack} and~\ref{fig:compare_gan_attack}).
\end{itemize}

\subsection{Robustness across Diverse Surrogate Models}
\label{sec:surrogate_robust}

Most link prediction attacks operate on a common principle: they all exploit feature similarity and structural similarity (embedding similarity also results from feature and structural similarity). This shared foundation suggests that a defense trained against one type of surrogate model should be robust against others.
To verify this hypothesis, we evaluate our defense against surrogate models built with three distinct prediction heads: a cosine similarity-based predictor (COS), an inner product-based predictor (IP), and a multilayer perceptron-based predictor (MLP). In our experiment, the defender trains the defense graph via a surrogate model with one specific predictor. The attacker, however, is free to use a model with any of the three predictors, creating both matched (defender and attacker use the same model) and mismatched scenarios.
As shown in Fig.~\ref{fig:results_parameter_sim}, our defense remains effective even in mismatched scenarios, with only slight performance degradation compared with when the models are matched.

\begin{figure}[t]
    \centering
    \includegraphics[width=\linewidth]{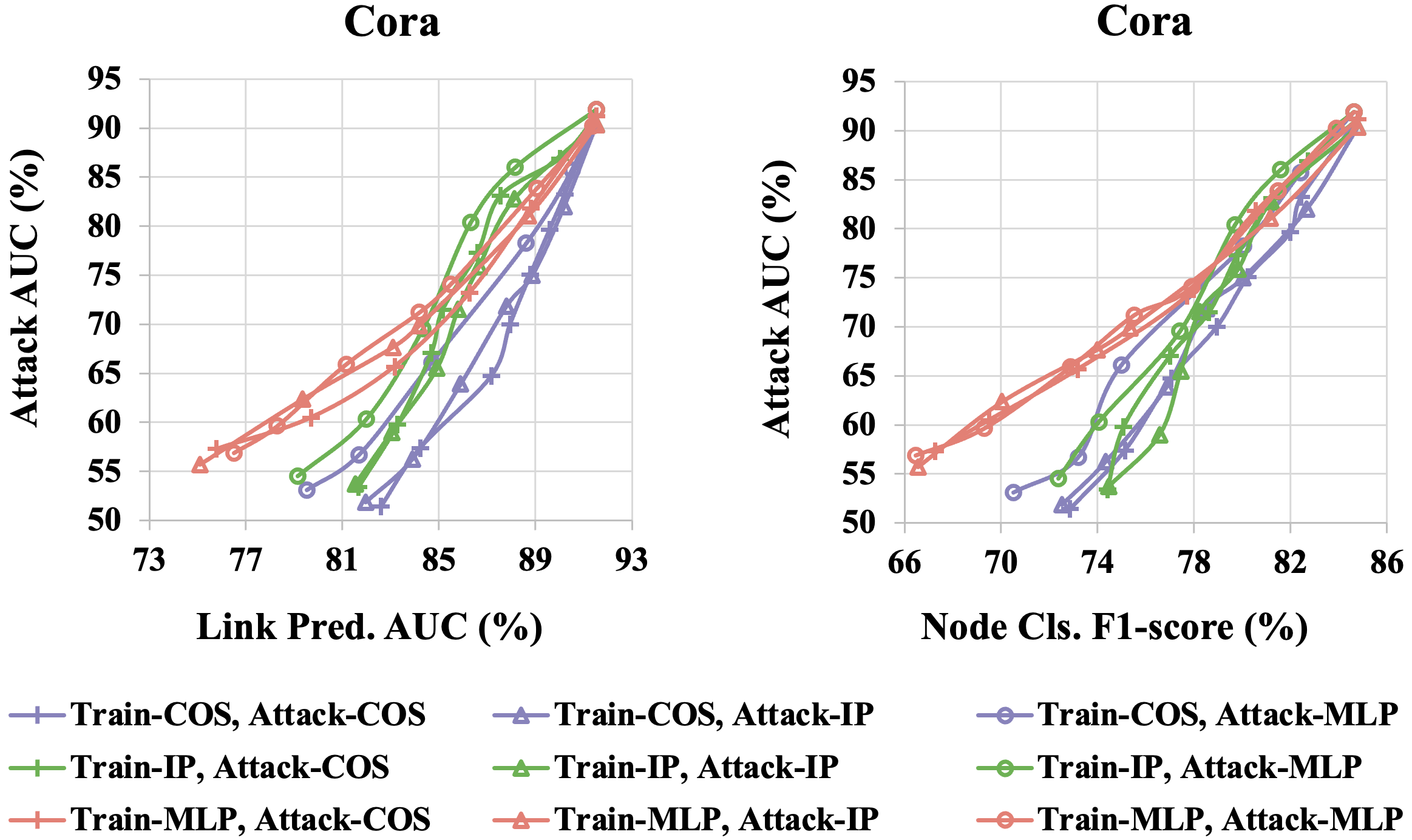}
    \vspace{-1em}
    \caption{Robustness across diverse surrogate models on \textit{Cora}. COS: cosine similarity-based predictor, IP: inner production-based predictor, MLP: multilayer perceptron-based predictor.}
    \label{fig:results_parameter_sim}
    \vspace{-1em}
\end{figure}

\end{document}